\documentclass[11pt,letterpaper]{article}
\pdfoutput=1
\usepackage{latexsym}
\usepackage{amsthm}
\usepackage{amsfonts}
\usepackage{amssymb}
\usepackage{paralist}
\usepackage{amsmath}
\usepackage{xspace}
\usepackage{mathrsfs}
\usepackage[letterpaper,margin=3cm]{geometry}
\usepackage{changepage}
\usepackage{appendix}
\usepackage[dvipsnames,usenames]{xcolor}
\usepackage{graphicx}
\usepackage[boxed,figure,vlined]{algorithm2e}
\usepackage{enumerate}
\usepackage{tabularx}

\definecolor{darkred}{rgb}{0.75,0,0}

\definecolor{darkgreen}{rgb}{0,0.5,0}
\usepackage{hyperref}
\hypersetup{
    unicode=false,          
    colorlinks=true,        
    linkcolor=red,          
    citecolor=darkgreen,        
    filecolor=magenta,      
    urlcolor=cyan           
}
\usepackage[capitalize, nameinlink]{cleveref}
\crefname{theorem}{Theorem}{Theorems}
\crefname{lemma}{Lemma}{Lemmas}

\usepackage{subcaption}
\usepackage{caption}
\usepackage{color}

\providecommand{\keywords}[1]{\textbf{Keywords:} #1}

\newcommand{\E}{\ensuremath{\mathbb{E}}}
\newcommand{\set}[1]{\left\{#1\right\}}

\newcommand{\eps}{\varepsilon}

\newcommand{\lev}{\ell}
\newcommand{\cost}{\ensuremath{\mathit{cost}}}
\newcommand{\latency}{L}
\newcommand{\child}{\ensuremath{\mathit{child}}}
\newcommand{\ALG}{\ensuremath{\mathsf{ALG}}}
\newcommand{\treelength}{\delta}

\newcommand{\arrow}{\ensuremath{\mathsf{Arrow}}}
\newcommand{\ARW}{\ensuremath{\mathcal{A}}}

\newcommand{\OPT}{\ensuremath{\mathcal{O}}}

\newcommand{\mtntotalcost}{\ensuremath{\mathit{C}}}
\newcommand{\mtncost}{\ensuremath{\mathit{c}}}
\newcommand{\mtn}{\scalebox{.5}[.5]{\ensuremath{\mathcal{M}}}}
\newcommand{\mc}{\mtncost_{\mtn}}
\newcommand{\manhattan}{\ensuremath{\mathrm{Manhattan}}}
\newcommand{\stt}{\ensuremath{S^*}}
\newcommand{\starw}{\mathbb{S}}
\newcommand{\Iarw}{\mathbb{I}}
\newcommand{\I}{I^*}
\newcommand{\nextt}{next}

\newcommand{\asy}{\scalebox{.55}[.55]{\ensuremath{\mathbb{A}}}}
\newcommand{\asyorder}{\pi_{\ARW}^{\asy}}

\newcommand{\hide}[1]{}
\renewcommand{\include}{\input}

\newcommand{\para}[1]{\medskip\noindent\textbf{#1}}
\newcommand{\paranospace}[1]{\noindent\textbf{#1}}

\newtheorem{theorem}{Theorem}[section]
\newtheorem{lemma}[theorem]{Lemma}

\newtheorem{corollary}[theorem]{Corollary}
\newtheorem{definition}{Definition}[section]
\newtheorem{observation}[theorem]{Observation}
\newtheorem{remark}{Remark}[section]

\title{Dynamic Analysis of the Arrow Distributed Directory Protocol in
  General Networks}
\author{Abdolhamid Ghodselahi and Fabian Kuhn\\
Department of Computer Science, University of Freiburg\\
  79110 Freiburg, Germany\\
{\small \texttt{\{hghods,kuhn\}@cs.uni-freiburg.de}}}

\date{}

\begin{document}

\maketitle              

\begin{abstract}
  The \ensuremath{\mathsf{Arrow}} protocol is a simple and elegant protocol to coordinate
  exclusive access to a shared object in a network. The protocol
  solves the underlying distributed queueing problem by using path
  reversal on a pre-computed spanning tree (or any other tree topology
  simulated on top of the given network).

  It is known that the \ensuremath{\mathsf{Arrow}} protocol solves the problem with a
  competitive ratio of $O(\log D)$ on trees of diameter $D$. This
  implies a distributed queueing algorithm with competitive ratio
  $O(s\cdot \log D)$ for general networks with a spanning tree of diameter
  $D$ and stretch $s$. In this work we show that when running the
  \ensuremath{\mathsf{Arrow}} protocol on top of the well-known probabilistic tree embedding
  of Fakcharoenphol, Rao, and Talwar [STOC 03], we obtain a randomized
  distributed queueing algorithm with a competitive ratio of $O(\log n)$
  even on general network topologies. The result holds even if the
  queueing requests occur in an arbitrarily dynamic and concurrent
  fashion and even if communication is asynchronous. From a technical
  point of view, the main of the paper shows that the competitive
  ratio of the \ensuremath{\mathsf{Arrow}} protocol is constant on a special family of tree
  topologies, known as hierarchically well separated trees. 
\end{abstract}

\section{Introduction}
\label{sec:intro}

Coordinating the access to shared data is a fundamental task that is
at the heart of almost any distributed system. For example, when
implementing a distributed shared memory system on top of a message
passing system, each shared register has to be kept in a coherent
state despite possibly a large number of concurrent requests to read
or write the shared register. In a distributed transactional memory
system, each transaction might need to operate on several shared
objects, which need to be kept in a consistent state
\cite{herlihy2007distributed,sharma2014distributed,zhang2010dynamic}.
When implementing a shared object on top of large-scale network, a
\emph{distributed directory protocol} can be used to improve
scalability of the system
\cite{abraham2004lls,attiya2010provably,awerbuch1995,chaiken90,demirbas2004hierarchy,herlihy2007distributed,sharma2014distributed}.
When a network node requires access to a shared object, the directory
moves a copy of the object to the node requesting the object. If the
node changes the state of the shared object, the directory protocol
has to make sure that all existing copies of the object are kept in a
consistent state.

\para{Distributed Queueing:}
At the core of many distributed directory implementations is the
following basic \emph{distributed queueing problem} that allows to
order potential concurrent access requests to a shared object
\cite{herlihy2001competitive}. The nodes of a network issue queueing
requests (e.g., requests to access a shared object) in a completely
dynamic and possibly arbitrarily concurrent manner. A queueing
protocol needs to globally order all the requests so that they can be
acted on consecutively. Formally, each request has to find its
\emph{predecessor request} in the order. That is, when enqueueing a
request $r$ issued by some node $v$, a queueing protocol needs to find
the request $r'$ that currently forms the tail of the queue and inform
the node $v'$ of request $r'$ about the new request $r$.

\para{The Arrow Protocol:}
A particularly simple and elegant solution for this distributed
queueing problem is given by the \arrow \ protocol, which was
introduced by Raymond in the context of distributed mutual exclusion
\cite{raymond1989tree}.  The \arrow \ protocol
operates on a directed tree topology $T=(V,E)$. In a quiescent state,
the tree is rooted at the node $u$ of the current tail of the queue,
i.e., all edges of $T$ are directed towards $u$. When a new queueing
request is issued at a node $v$, the direction of the edges on the
path between $v$ and the previous tail $u$ is reversed so that the
tree is now rooted at $v$. For a precise description of the protocol,
we refer to \Cref{sec:model}. It has been shown in
\cite{demmer1998arrow} that the \arrow\ protocol correctly solves the
queueing problem even in an asynchronous system even if the requests
are issued in a completely dynamic and possibly concurrent
way. Moreover, the \arrow\ protocol guarantees that every request
finds the node of its predecessor on a direct path (i.e., within $D$
time units if $D$ is the diameter of $T$). In
\cite{herlihy2006dynamic}, it was further shown that on a tree $T$,
the overall cost of the \arrow \ protocol for ordering a dynamic set
of queueing requests is within a factor $O(\log D)$ of the cost of an
optimal offline queueing algorithm, which knows the request sequence
in advance.\footnote{Note that this implies a competitive ratio of
  $O(s \cdot \log D)$ for general graphs if a spanning tree $T$ of
  diameter $D$ and stretch $s$ is given.}
  
\para{Contribution:}
In the present paper, we strengthen the result of
\cite{herlihy2006dynamic} and we show that when run on the right
underlying tree, the \arrow \ protocol is $O(\log n)$-competitive even
on general network topologies. The best previously known competitive
ratio for the distributed queueing problem with arbitrarily
dynamically injected requests on general graphs is
$O(\log^2 n \cdot \log D)$ as shown in \cite{sharma2015analysis} for
the hierarchical schemes defined of
\cite{awerbuch1995,sharma2014distributed}. This shows that (under some
assumptions), the simple and elegant \arrow \ protocol outperforms all
existing significantly more complicated distributed queueing
protocols.\footnote{Our protocol is based on a randomized
  tree construction and its competitive ratio is w.r.t.\ an oblivious
  adversary. Other protocols with polylogarithmic competitive ratio
  are deterministic and they therefore also work in the presence of an
  adaptive adversary.} For a more detailed comparison of our results with
existing protocols, we refer to the discussion in \Cref{sec:relatedwork}.

More specifically, as our main technical result, we show that the
\arrow\ protocol is $O(1)$-competitive when it is run on a special class
of trees known as \emph{hierarchically well separated trees}
\cite{bartal96}. A hierarchically well separated tree (in the
following referred to as an HST) is a weighted, rooted tree where on
each level, all the nodes are at the same distance to the root and all
the leaves are on the same level (and thus also at the same distance
to the root). Further, the edge lengths decrease exponentially (by a constant factor
per level) when going from the root towards the leaves. When running
\arrow\ on an HST $T$, we assume that all requests are issued at the
leaves of $T$. We show that the total cost of an \arrow\ execution on an
HST $T$ is within a constant factor of the total cost of an optimal
offline algorithm for the given set of requests. Our result even holds
if the communication on $T$ is asynchronous.

\begin{theorem}\label{thm:HSTmain}
  Assume that we are given an HST $T$ with parameter $2$ and queueing
  requests $R$ that arrive in an arbitrarily dynamic manner at the
  leaves of $T$. When using the \arrow \ protocol on tree $T$, the
  total cost for ordering the requests in $R$ is within a constant
  factor of the cost of an optimal offline algorithm for ordering the
  requests $R$ on $T$. This even holds if communication is
  asynchronous.
\end{theorem}

\begin{remark}
  Because the statement of the theorem applies to the general
  asynchronous case, it also captures a synchronous scenario, where
  the delay on each edge is fixed, but might be smaller than the
  actual weight of the edge in the HST. Note that such executions are
  relevant because an HST is often built as an overlay graph on top of
  an underlying network graph $G$ and the delay of simulating a single
  HST edge might be smaller than the weight of the edge.
\end{remark}

For a precise description of the \arrow \ protocol and the definition
of queueing cost, we refer to \Cref{sec:model}. When combining
\Cref{thm:HSTmain} with the celebrated probabilistic tree
embedding of Fakcharoenphol, Rao, and Talwar
\cite{fakcharoenphol2003tight}, we get our main result for general
graphs. In \cite{fakcharoenphol2003tight}, it is shown that there is a
randomized algorithm that given an arbitrary $n$-point metric $(X,d)$
constructs an HST $T$ such all points $X$ are mapped to leaves of $T$,
all distances in $(X,d)$ are upper bounded by the respective distances
in $T$, and the expected distance between any two leaves in $T$ is
within an $O(\log n)$ factor of the distance between the corresponding
two points in $X$. When constructing such an HST $T$ for a given graph
$G$ and when assuming an oblivious adversary\footnote{That is, when
  assuming that the sequence of requests is statistically independent
  of the randomness used to construct the HST $T$.}, this implies that
the expected total cost of \arrow \ on $T$ is within an $O(\log n)$
factor of the optimal offline queueing cost on $G$. We also note that
an efficient distributed construction of the HST embedding of
\cite{fakcharoenphol2003tight} has been given in
\cite{ghaffari2014near}.

\begin{theorem}\label{thm:arrowmain}
  Assume that we are given an arbitrary graph $G=(V,E)$ and queueing
  requests $R$ that arrive in an arbitrarily dynamic manner at the
  nodes of $G$. There is a randomized construction of an HST $T$ that
  can be simulated on $G$ such that when running \arrow \ on $T$, we get
  a distributed queueing algorithm for $G$ with competitive ratio at
  most $O(\log n)$ against an oblivious adversary providing the
  sequence of requests. This even holds if communication is
  asynchronous.
\end{theorem}

\para{Organization of the Paper:} The remainder of the paper is
organized as follows. \Cref{sec:model} formally defines the
queueing problem, the \arrow \ protocol, as well as the cost model
used in our paper. The section also contains some lemmas that
establish some basic properties that are needed for the rest of the
paper.  \Cref{sec:treeanalysis} analyzes the cost of an optimal
offline algorithm on an HST $T$ by relating it to the total weight of
an MST defined on the set of requests. In \Cref{sec:arrowanalysis},
we introduce a general framework to analyze
the queueing cost of distributed queueing algorithms on an HST $T$ and
the framework is applied to synchronous executions of the \arrow\
protocol. The analysis of asynchronous executions appears in \Cref{sec:asynchmodel}.

\subsection{Related Work}
\label{sec:relatedwork}

The \arrow\ protocol has been introduced by Raymond
\cite{raymond1989tree} as a way to solve the mutual exclusion problem
in a network. The protocol was later reinvented by Demmer and Herlihy
\cite{demmer1998arrow}, who used \arrow\ to implement a distributed
directory \cite{chaiken90}. Over the years, \arrow\ has been used and
analyzed in different contexts
\cite{herlihy1999aleph,herlihy2006dynamic,herlihy2001ordered,herlihy1999tale,peleg1999variant,tirthapura2006self}. The
protocol has been implemented as a part of Aleph Toolkit
\cite{herlihy1999aleph} and shown to outperform centralized schemes
significantly in practice \cite{herlihy1999tale}. Several other
tree-based distributed queueing protocols that are similar to the
\arrow\ protocol have also been proposed in the literature. A protocol
that combines the ideas of \arrow\ with path compression has been
implemented in the Ivy system \cite{li1989memory}. The amortized cost
to serve a single request is only $O(\log n)$ \cite{ginat1989tight},
however the protocol needs a complete graph as the underlying network
topology. There are also other similar protocols that operate on fixed
trees. The Relay protocol \cite{zhang2010dynamic} has been introduced as a
distributed transactional memory protocol. It is run on top of a fixed
spanning tree similar to \arrow, however to more efficiently deal with
aborted transactions, it does not always move the shared object to the
node requesting it. Further, in \cite{attiya2010provably}, a
distributed directory protocol called Combine has been
proposed. Combine runs on a fixed overlay tree and it is in particular
shown in \cite{attiya2010provably} that Combine is starvation-free.

The first paper to study the competitive ratio of concurrent
executions of a distributed queueing protocol is
\cite{herlihy2001competitive}. The paper shows that in synchronous
executions of \arrow\ on a tree $T$, if all requests are issued at
time $0$ (known as one-shot executions), the total cost of \arrow\ is
within a factor $O(\log |R|)$ compared with the optimal queueing cost
on tree $T$. This analysis has later been extended (and slightly
strengthened) to the general concurrent setting where requests are
issued in an arbitrarily dynamic fashion. In \cite{herlihy2006dynamic},
it is shown that in this case, the total cost of \arrow\ is within a
factor $O(\log D)$ of the optimal cost on the tree $T$. Later,
the same bounds have also been proven for the Relay protocol
\cite{zhang2010dynamic} and the Combine protocol
\cite{attiya2010provably}. Typically, these protocols are run on a
spanning tree or an overlay tree on top of an underlying general
network topology. While the cost of all these protocols is small when
compared with the optimal queueing cost on the tree, the cost of the
protocols might be much larger when compared with the optimal cost on
the underlying topology. In this case, the competitive ratio becomes
$O(s\cdot \log D)$, where $s$ is the stretch of the tree. There are
underlying graphs (e.g., cycles) for which every spanning tree and
even every overlay tree has stretch $\Omega(n)$
\cite{gupta2001steiner,RR98}. The fact that even the best spanning tree might
have large stretch initiated the work on distributed queueing
protocols that run on more general hierarchical structures. In
\cite{herlihy2007distributed}, a protocol called Ballistic is
introduced and analyzed for the sequential and the one-shot
case. Ballistic has competitive ratio $O(\log D)$, however the
protocol requires the underlying distance metric to have bounded
doubling dimension and it thus cannot be applied in general
networks. The best protocol known for general networks is Spiral,
which was introduced in \cite{sharma2014distributed}. Spiral is based
on a hierarchy of overlapping clusters that cover the graph. It's
general structure is thus somewhat resembling the classic sparse
partitions and mobile objects solutions by Awerbuch and Peleg
\cite{awerbuch90,awerbuch1995}. The competitive ratio of Spiral is
shown to be $O(\log^2 n \cdot \log D)$ for sequential and one-shot
executions in \cite{sharma2014distributed}. In
\cite{sharma2015analysis}, a general framework to analyze the cost of
concurrent executions of hierarchical queueing and directory protocols
has been presented. In particular, in \cite{sharma2015analysis}, the
competitive analysis of Spiral and also of the classic mobile object
algorithm of Awerbuch and Peleg \cite{awerbuch90,awerbuch1995} has
been extended to the dynamic setting. In \cite{herlihy2006dynamic}, a sketch is given of how the
competitive analysis for \arrow\ generalized to the asynchronous case.


\section{Model, Problem Statement, and Preliminaries}
\label{sec:model}

\paranospace{Communication Model:} We consider a standard message
passing model on a network modeled by a graph $G=(V,E)$. In some
cases, the edges of $G$ have weights $w: E\to \mathbb{R}_{>0}$, which
are assumed to be normalized such that $w(e)\geq 1$ for all $e\in
E$.
We distinguish between synchronous and asynchronous executions. In a
\emph{synchronous execution}, the delay for sending a message from a
node $u$ to a node $v$ over an edge $e$ connecting $u$ and $v$ is
exactly $1$ if the edge is unweighted and exactly $w(e)$ otherwise. In
an \emph{asynchronous execution}, message delays are arbitrary,
however when analyzing an asynchronous execution, we assume that the
message delay over an edge $e$ is upper bounded by the edge weight
$w(e)$ (or by $1$ in the unweighted case).

\para{The Distributed Queueing Problem:} In the \emph{distributed
  queueing problem} on a graph $G=(V,E)$, a set $R$ of queueing
requests $r_i=(v_i,t_i)$ are issued at the nodes of $V$ in an
arbitrarily dynamic fashion. The goal of a queueing algorithm is to
order all the requests. Specifically, if a request $r_i=(v_i,t_i)$ is
issued at node $v_i$ at time $t_i\geq 0$, the algorithm needs to
enqueue the request $r_i$ by informing the node $v_j$ of the
predecessor request $r_j=(v_j,t_j)$ in the constructed global
order. For this purpose, every queueing algorithm in particular has to
send (possibly indirectly) a respective message from node $v_i$ to
$v_j$. We assume that at time $0$, when an execution starts, the tail
of the queue is at a given node $v_0\in V$. Formally, this is modeled
as a request $r_0=(v_0,0)$ which has to be ordered first by any
queueing protocol. We sometimes refer to $r_0$ as the dummy
request. For a set $R'$ of queueing request (and sometimes by
overloading notation also for a set of request indexes), we define
$t_{\min}(R')$ and $t_{\max}(R')$ to be the minimum and the maximum
issue time $t$ of any request $r=(v,t)\in R'$, respectively.

\para{The Arrow Protocol:} The \arrow\ protocol \cite{raymond1989tree}
is a distributed queueing protocol that operates on a tree network
$T=(V,E)$. At each point in time, each node $v\in V$ has exactly one
outgoing link (arrow) pointing either to one of the neighbors of $v$
or to the node $v$ itself. In a quiescent state, the arrow of the node
of the request at the tail of the queue points to itself and all other
arrows point towards the neighbor on the path towards the tail of the
queue (i.e., the tree is directed towards the current tail). When a
new request at a node $v\in V$ occurs, a ``find predecessor'' message
is sent along the arrows until it finds the predecessor request. While
following the path to the direction of the arrows are reversed. More
formally, a request $r$ at node $v$ is handled as follows.

\begin{enumerate} 
\item If the arrow of $v$ points to $v$ itself, $r$
  is queued directly behind the previous request issued at $v$.
  Otherwise if the arrow points to neighbor $u$, \emph{atomically}, a
  ``find predecessor'' message (including the information about
  request $r$) is sent to $u$ and the arrow of $v$ is redirected to
  $v$ itself. 
\item If a node $u$ receives a ``find predecessor''
  message for request $r$ from a neighbor $w$, if the arrow of $u$
  points to itself, \emph{atomically}, the request $r$ is queued
  directly behind the last request issued by node $u$ and the arrow of
  $u$ is redirected to node $w$. Otherwise, if the arrow of $u$ points
  to neighbor $x$, \emph{atomically}, the ``find predecessor'' message
  is forwarded to node $x$ and the arrow of node $u$ is redirected to
  node $w$.
\end{enumerate}

For a more detailed description of the \arrow \ protocol and of how
\arrow\ handles concurrent requests, we refer the reader to
\cite{demmer1998arrow,herlihy2006dynamic}. It was shown in
\cite{demmer1998arrow} that the \arrow \ protocol correctly orders a
given sequence of requests even in an asynchronous network. Moreover
as shown in \cite{demmer1998arrow,herlihy2006dynamic}, when operating
on tree $T$, the protocol always finds the predecessor of a request on
the direct path on $T$. As a result, if two requests $r'$ and $r$ are
at distance $d$ on $T$ and if $r'$ is the predecessor of $r$ in the
queueing order, the ``find predecessor'' message initiated by request
$r$ finds the node of request $r'$ in time exactly $d$ in the
synchronous setting and in time at most $d$ in the asynchronous
model. Further, it is shown in \cite{herlihy2006dynamic} that the
successor request of a request $r$ at node $v$ in the queue is always
the remaining request $r''$ that first reaches $v$ on a direct
path. This ``greedy'' nature of the \arrow \ ordering was used in
\cite{herlihy2001competitive}, where it was shown that in the one-shot
case when all requests occur at time $0$, the \arrow \ order
corresponds to a greedy (nearest neighbor) TSP path through requests,
whereas an optimal offline algorithm corresponds to an optimal TSP
path on the request set. The competitive ratio on trees then follows
from the fact that the nearest neighbor heuristic provides a
logarithmic approximation of the TSP problem
\cite{nearestneighbor}. In \cite{herlihy2006dynamic}, this analysis
was extended and it was shown that even in the fully dynamic case, it
is possible to reduce the problem to a (generalized) TSP nearest
neighbor analysis. Formally, the greedy nature of the \arrow\ protocol
in the synchronous setting is captured by \Cref{le:timewindow} in
\Cref{sec:treeanalysis}, whereas the corresponding property in
the asynchronous setting is formally discussed in \Cref{sec:asynchmodel}.


\para{Hierarchically Well Separated Trees:} The notion of a
\emph{hierarchically well separated tree} (HST) was defined by Bartal
in \cite{bartal96}. Given a parameter $\alpha>1$, an HST of depth $h$
is a rooted tree with the following properties. All children of the
root are at distance $\alpha^{h-1}$ from the root. Further, every
subtree of the root is an HST of depth $h-1$ that is characterized by
the same parameter $\alpha$ (i.e., the children 2 hops away from the
root are at distance $\alpha^{h-2}$ from their parents). The
probabilistic tree embedding result of \cite{fakcharoenphol2003tight}
shows that for every metric space $(X,d)$ with minimum distance
normalized to $1$ and for every constant $\alpha> 1$, there is a
randomized construction of an HST $T$ with a bijection $f$ of the
points in $X$ to the leaves of $T$ such that for every $x,y\in X$,
$d(x,y)\leq d_T(f(x),f(y))$ and such that the expected tree distance
$\E\big[d_T(f(x),f(y))\big] = O(\log |X|)\cdot d(x,y)$. Further, an
efficient distributed implementation of the construction of
\cite{fakcharoenphol2003tight} for the distances of a given network graph
was given in \cite{ghaffari2014near}.

The main technical result of this paper is an analysis of \arrow\
on an HST $T$ if all requests are issued at leaves of $T$. Throughout
the paper, the HST parameter $\alpha$ is set to $\alpha =2$.  For
convenience, we number the levels of an HST $T$ of depth $h$ from $0$
to $h$, where the level $0$ nodes are the leaves and the single level
$h$ node is the root. For $\ell\in\set{0,\dots,h}$,
$\delta(\ell):=2^{\ell+1}-2$ denotes the distance between two leaves
for which the least common ancestor is on level $\ell$.

\para{Cost Model:} Assume when applying some queueing algorithm $\ALG$
to the dynamic set of request $R$, the requests are ordered according
to the permutation $\pi_{\ALG}$ such that the request ordered at
position $i$ in the order is $r_{\pi_{\ALG(i)}}$. For every
$i\in \set{1,\dots,|R|-1}$, we define the \emph{cost of ordering
  $r_{\pi_{\ALG}(i)}$ after $r_{\pi_{\ALG}(i-1)}$} as the time it takes a
queueing algorithm to enqueue the request $r_{\pi_{\ALG}(i)}$ as the
successor of $r_{\pi_{\ALG}(i-1)}$. More specifically, we assume that
request $r_{\pi_{\ALG}(i)}$ can be enqueued as soon as the predecessor
request $r_{\pi_{\ALG}(i-1)}$ is in the system and as soon as node
$v_{\pi_{\ALG}(i-1)}$ knows about request $r_{\pi_{\ALG}(i)}$. Assume that
algorithm $\ALG$ informs node $v_{\pi_{\ALG}(i-1)}$ (through a message)
about $r_{\pi_{\ALG}(i)}$ at time $t_{\ALG}(i)$. The cost (latency)
$\latency_{\mathsf{ALG}}(r_{\pi_{\ALG}(i-1)},r_{\pi_{\ALG}(i)})$ incurred
for enqueueing request $r_{\pi_{\ALG}(i)}$ and the overall cost
(latency) $\cost_{\ALG}$ of \ALG\ are then defined as follows.
\begin{eqnarray}
  \latency_{\ALG}(r_{\pi_{\ALG}(i-1)},r_{\pi_{\ALG}(i)})
  & := &
         \max\set{t_{\ALG}(i),t_{\pi_{\ALG}(i-1)}}-t_{\pi_{\ALG}(i)},\label{eq:latencycost}\\
  \cost_{\ALG}(\pi_{\ALG})
  & := &
         \sum_{i=1}^{|R|-1}\latency_{\ALG}(r_{\pi_{\ALG}(i-1)},r_{\pi_{\ALG}(i)}).\label{eq:totalcost}
\end{eqnarray}
We next specify the above cost more concretely for \arrow \ and for an
optimal offline algorithm. Assume that we have an execution \ARW\ of
the \arrow \ protocol that operates on a tree $T$. Let $\pi_{\ARW}$ be
the ordering induced by the \arrow\ execution \ARW. When the ``find
predecessor'' message of a request $r_{\pi_{\ARW}(i)}$ arrives at the
node of the predecessor request $r_{\pi_{\ARW}(i-1)}$, clearly the
request $r_{\pi_{\ARW}(i-1)}$ has already occurred and thus we always
have
$\latency_{\ARW}(r_{\pi_{\ARW}(i-1)},r_{\pi_{\ARW}(i)}) =
t_{\ARW}(i)-t_{\pi_{\ARW}(i)}$
for any \arrow\ execution. Further note, that in a synchronous
execution of arrow on tree $T$, because \arrow\ always finds the
predecessor on the direct path, this latency cost is always equal to
the distance between the respective nodes in $T$.

When studying in the cost of an optimal offline queueing algorithm
\OPT, we assume that \OPT\ knows the whole sequence of requests in
advance. However, \OPT\ still needs to send messages from each request
to its predecessor request. The message delays are not under the
control of the optimal offline algorithm. When lower bounding the cost
of \OPT, we can therefore assume that all communication is synchronous
even in the asynchronous case. Note that a synchronous execution is a
possible strategy of the asynchronous scheduler. When operating on a
graph $G$, the latency cost of \OPT\ for ordering a request $r_j$ as
the successor of a request $r_i$ is then exactly
$\latency_{\OPT}^G(r_i,r_j) = \max\set{t_i - t_j, d_G(v_i,v_j)}$. As
we analyze \arrow\ on an HST $T$ that is simulated on top of an
underlying network $G$, we directly define the optimal offline w.r.t.\
synchronous executions on the tree $T$ as follows.
\begin{eqnarray}
  \latency_{\OPT}^T(r_{\pi_{\OPT}^T(i-1)},r_{\pi_{\OPT}^T(i)})
  & := & 
         \max\set{d_T(v_{\pi_{\OPT}^T(i-1)},v_{\pi_{\OPT}^T(i)}),t_{\pi_{\OPT}^T(i-1)}-t_{\pi_{\OPT}^T(i)}},\label{eq:optlatency}\\
  \cost_{\OPT}^T(\pi_{\OPT})
  & := &
         \sum_{i=1}^{|R|-1}\latency_{\OPT}^T(r_{\pi_{\OPT}^T(i-1)},r_{\pi_{\OPT}^T(i)}).\label{eq:opttotal}
\end{eqnarray}
The ordering $\pi_{\OPT}$ is chosen such that the total cost
$\cost_{\OPT}^T(\pi_{\OPT})$ in \eqref{eq:opttotal} is minimized. The next lemma
shows that when using the randomized HST construction of
\cite{fakcharoenphol2003tight}, the cost \eqref{eq:opttotal} is within
a logarithmic factor of the optimal offline cost on the underlying
network graph $G$.

\begin{lemma} \label{lemma:FRTapproximation}
  Assume $T$ is an HST that is constructed on top of an $n$-node
  network graph $G$ by using the randomized algorithm of
  \cite{fakcharoenphol2003tight} and assume that there is a dynamic
  set of queueing requests issued at the nodes of $G$. If the sequence
  of requests is independent of the randomness of the randomized HST
  construction, the expected optimal total cost on $T$ (as defined in
  \eqref{eq:opttotal}) is within a factor $O(\log n)$ of the optimal
  offline queueing cost on $G$.
\end{lemma}
\begin{proof}
  Let $\pi_{\OPT}^G$ and $\pi_{\OPT}^T$ be the optimal orderings
  w.r.t.\ the optimal offline costs $\latency_{\OPT}^G(r_i,r_j)$ and
  $\latency_{\OPT}^T(r_i,r_j)$ on $G$ and $T$, respectively, as
  defined above. We have
  \begin{eqnarray*}
    \E\left[\cost_{\OPT}^T(\pi_{\OPT}^T)\right]
    & = & 
          \E\left[\sum_{i=1}^{|R|-1}\latency_{\OPT}^T(r_{\pi_{\OPT}^T(i-1)},r_{\pi_{\OPT}^T(i)})\right]\\
    & \leq &\sum_{i=1}^{|R|-1}
             \E\left[\latency_{\OPT}^T(r_{\pi_{\OPT}^G(i-1)},r_{\pi_{\OPT}^G(i)})\right]\\
    & = & \sum_{i=1}^{|R|-1}\E\left[
          \max\set{d_T(v_{\pi_{\OPT}^G(i-1)},v_{\pi_{\OPT}^G(i)}),t_{\pi_{\OPT}^G(i-1)}-t_{\pi_{\OPT}^G(i)}}
          \right]\\
    & \leq &
             2\cdot \sum_{i=1}^{|R|-1}\max\set{
             \E\left[d_T(v_{\pi_{\OPT}^G(i-1)},v_{\pi_{\OPT}^G(i)})\right], t_{\pi_{\OPT}^G(i-1)}-t_{\pi_{\OPT}^G(i)}
             }\\
    & \leq & 2\cdot\sum_{i=1}^{|R|-1}
             \max\set{ O(\log n)\cdot
             d_G(v_{\pi_{\OPT}^G(i-1)},v_{\pi_{\OPT}^G(i)}),
             t_{\pi_{\OPT}^G(i-1)}-t_{\pi_{\OPT}^G(i)}}\\
    & \leq & O(\log n)\cdot \sum_{i=1}^{|R|-1}
             \max\set{ d_G(v_{\pi_{\OPT}^G(i-1)},v_{\pi_{\OPT}^G(i)}),
             t_{\pi_{\OPT}^G(i-1)}-t_{\pi_{\OPT}^G(i)}}\\
    & \leq & O(\log n) \cdot \cost_{\OPT}^G(\pi_{\OPT}^G).
  \end{eqnarray*}
  The first inequality follows from the fact that $\pi_{\OPT}^T$ is an
  optimal ordering w.r.t.\ the cost $\latency_{\OPT}^T(r_i,r_j)$ and
  by linearity of expectation. The second inequality follows because
  for every non-negative random variable $X$ and every fixed (possibly
  negative) constant $c$, it holds that
  $\E[\max\set{X,c}]\leq 2 \cdot \max\set{\E[X],c}$. The third
  inequality follows from the expected stretch bound of the HST
  construction of \cite{fakcharoenphol2003tight}, and the fourth
  inequality follows because for all values $\lambda\geq 1$, $a\geq 0$
  and $b\in \mathbb{R}$, it holds that
  $\max\set{\lambda a, b} \leq \lambda \cdot \max\set{a,b}$.
\end{proof}
\hide{
\begin{proof}
  Let $\pi_{\OPT}^G$ and $\pi_{\OPT}^T$ be the optimal orderings
  w.r.t.\ the optimal offline costs $\latency_{\OPT}^G(r_i,r_j)$ and
  $\latency_{\OPT}^T(r_i,r_j)$ on $G$ and $T$, respectively, as
  defined above. We have
  \begin{eqnarray*}
    \E\left[\cost_{\OPT}^T(\pi_{\OPT}^T)\right]
    & = & 
          \E\left[\sum_{i=1}^{|R|-1}\latency_{\OPT}^T(r_{\pi_{\OPT}^T(i-1)},r_{\pi_{\OPT}^T(i)})\right]\\
    & \leq &\sum_{i=1}^{|R|-1}
             \E\left[\latency_{\OPT}^T(r_{\pi_{\OPT}^G(i-1)},r_{\pi_{\OPT}^G(i)})\right]\\
    & = & \sum_{i=1}^{|R|-1}\E\left[
          \max\set{d_T(v_{\pi_{\OPT}^G(i-1)},v_{\pi_{\OPT}^G(i)}),t_{\pi_{\OPT}^G(i-1)}-t_{\pi_{\OPT}^G(i)}}
          \right]\\
    & \leq &
             2\cdot \sum_{i=1}^{|R|-1}\max\set{
             \E\left[d_T(v_{\pi_{\OPT}^G(i-1)},v_{\pi_{\OPT}^G(i)})\right], t_{\pi_{\OPT}^G(i-1)}-t_{\pi_{\OPT}^G(i)}
             }\\
    & \leq & 2\cdot\sum_{i=1}^{|R|-1}
             \max\set{ O(\log n)\cdot
             d_G(v_{\pi_{\OPT}^G(i-1)},v_{\pi_{\OPT}^G(i)}),
             t_{\pi_{\OPT}^G(i-1)}-t_{\pi_{\OPT}^G(i)}}\\
    & \leq & O(\log n)\cdot \sum_{i=1}^{|R|-1}
             \max\set{ d_G(v_{\pi_{\OPT}^G(i-1)},v_{\pi_{\OPT}^G(i)}),
             t_{\pi_{\OPT}^G(i-1)}-t_{\pi_{\OPT}^G(i)}}\\
    & \leq & O(\log n) \cdot \cost_{\OPT}^G(\pi_{\OPT}^G.
  \end{eqnarray*}
  The first inequality follows from the fact the $\pi_{\OPT}^T$ is an
  optimal ordering w.r.t.\ the cost $\latency_{\OPT}^T(r_i,r_j)$ and
  by linearity of expectation. The second inequality follows because
  for every non-negative random variable $X$ and every fixed (possibly
  negative) constant $c$, it holds that
  $\E[\max\set{X,c}]\leq 2 \cdot \max\set{\E[X],c}$. The third
  inequality follows from the expected stretch bound of the HST
  construction of \cite{fakcharoenphol2003tight}, and the fourth
  inequality follows because for all values $\lambda\geq 1$, $a\geq 0$
  and $b\in \mathbb{R}$, it holds that
  $\max\set{\lambda a, b} \leq \lambda \cdot \max\set{a,b}$.
\end{proof}
}

Given \Cref{thm:HSTmain} (which will be proven as the main
technical result of the paper) and \Cref{lemma:FRTapproximation},
we immediately get \Cref{thm:arrowmain}. We note in light of
the remark following the statement of \Cref{thm:HSTmain} in
\Cref{sec:intro}, the statement of \Cref{thm:arrowmain} is also true for synchronous executions on the
underlying graph $G$.

\para{Manhattan Cost:} In the dynamic competitive analysis of \arrow\
on general trees in \cite{herlihy2006dynamic}, it has been shown that
it is useful to study the optimal ordering w.r.t.\ to the following
\emph{Manhattan cost} on a tree $T$ between two queueing requests
$r_i=(v_i,t_i)$ and $r_j=(v_j,t_j)$.
\begin{equation}\label{eq:ManhattanCost}
  \mtncost_{\mtn}^T(r_i,r_j) := d_T(v_i,v_j) + |t_i - t_j|.
\end{equation}
As the cost function $\mtncost_{\mtn}(r_i,r_j)$ defines a metric space on the
request set, the problem of finding an optimal ordering w.r.t.\ the cost
$\mtncost_{\mtn}(r_i,r_j)$ is a metric TSP problem.\footnote{The
  relation of \arrow\ and the TSP problem was already exploited in
  \cite{herlihy2006dynamic} when analyzing \arrow\ on general trees.}
As a result, we will for example use that the total weight of an MST
on the set of request w.r.t.\ the weight function
$\mtncost_{\mtn}(r_i,r_j)$ is within a factor $2$ of the cost of an
optimal TSP path. The following definition is inspired by Lemma 3.12
in \cite{herlihy2006dynamic}.
\begin{definition}[Condensed Request Set]\label{def:condensed}
  A set $R$ of queueing requests $r_i=(v_i,t_i)$ on a tree $T$ is
  called \emph{condensed} if for any two requests $r_i=(v_i,t_i)$ and
  $r_j=(v_j,t_j)$ that are consecutive w.r.t.\ time of occurrence,
  there exits requests $r_a=(v_a,t_a)$ and $r_b=(v_b,t_b)$ such that
  $t_a\leq t_i$, $t_b\geq t_j$, and $d_T(v_a,v_b)\geq t_b-t_a$.
\end{definition}

It is shown in \cite{herlihy2006dynamic} that for condensed request
sets, the total optimal Manhattan cost is within a constant factor of
the optimal offline queueing cost.
\begin{lemma}[Lemma 3.17 in \cite{herlihy2006dynamic} rephrased] \label{lemma:manhattanopt}
  If the request set $R$ is condensed, then on any tree $T$ and for
  every ordering $\pi$ on the requests, it holds that
  \[
  \sum_{i=1}^{|R|-1} \mtncost_{\mtn}^T(r_{\pi(i-1)}, r_{\pi(i)}) \leq
  12\cdot \sum_{i=1}^{|R|-1}\latency_{\OPT}^T(r_{\pi(i-1)}, r_{\pi(i)}).
  \]
\end{lemma}

For synchronous executions on trees, it is also shown in
\cite{herlihy2006dynamic} that every request set $R$ can be
transformed into a condensed request set without changing the ordering
(and the cost) of \arrow\ and without increasing the optimal offline
cost.

\begin{lemma}[Lemma 3.11 in \cite{herlihy2006dynamic} rephrased] \label{le:transformation}
  Let $R$ be a set of queueing requests issued on a tree $T$ and let
  $r_i=(v_i,t_i)$ and $r_j=(v_j,t_j)$ be two requests of $R$ that are
  consecutive w.r.t.\ time of occurrence. Further, choose two requests
  $r_a=(v_a,t_a)$ with $t_a\leq t_i$ and $r_b=(v_b,t_b)$ with $t_b\geq
  t_j$ minimizing $\delta:=t_b-t_a-d_T(v_a,v_b)$. if $\delta>0$, every
  request $r=(v,t)$ with $t\geq t_j$ can be replaced by a request
  $r'=(v,t-\delta)$ without changing the synchronous \arrow\ order and without
  increasing the optimal offline cost.
\end{lemma}

\Cref{le:transformation} implies that every request set $R$ can
be transformed into a condensed set $R'$ without changing the
synchronous order of \arrow\ and without increasing the optimal
offline cost. For the analysis of \arrow\ in synchronous systems, we
can thus w.l.o.g.\ assume that the request set is condensed. In
\Cref{sec:asynchmodel}, we show that this also holds in
asynchronous systems.


\section{Analysis of the Optimal Offline Cost}
\label{sec:treeanalysis}

This and the next section discuss the main technical contribution of the paper
and analyzes the total cost of a synchronous \arrow\ execution when
run on an HST $T$. Throughout this section, we assume that a fixed HST
$T$, a set of dynamic requests $R$ placed at the leaves of $T$, and a
synchronous execution of \arrow\ with request set $R$ on $T$ are
given. For convenience, we relabel the requests in $R$ so that they
are ordered according to the queueing order resulting from the given
\arrow\ execution on $T$. That is, we assume that for all
$i\in\set{0,\dots,|R|-1}$, request $r_i=(v_i,t_i)$ is the
$i^{\mathit{th}}$ request in \arrow's order. Note that $r_0=(v_0,0)$
is still the dummy request defining the initial tail of the queue.  As
discussed in \Cref{sec:model}, the \arrow\ order can be seen as
a greedy ordering in the following sense. Given the first $i-1$
requests in the order, the $i^{\mathit{th}}$ request $r_i$ is a
request $r=(v,t)$ from the subset of the remaining requests that can
reach the node $v_{i-1}$ of request $r_{i-1}$ first immediately
sending a message at time $t$ from node $v$ to node $v_{i-1}$. This
greedy behavior is captured by the following basic lemma. The
generalization of this basic greedy property to the asynchronous
setting is discussed in \Cref{sec:asynchmodel}. For a more
thorough discussion, we also refer to \cite{herlihy2006dynamic}.

\begin{lemma}\label{le:timewindow}
  Consider a synchronous execution of \arrow\ on tree $T$ and consider
  two arbitrary requests $r_i$ and $r_j$ for which $1\leq i<j$ (i.e.,
  $r_j$ is ordered after $r_i$ by \arrow). Then it holds that
    \vspace*{-2mm}
  \begin{enumerate}
  \item $t_i + d_T(v_{i-1},v_i) \leq t_j + d_T(v_{i-1},v_j)$ and
  \item $t_i \leq t_j + d_T(v_i,v_j)$.  
  \end{enumerate}
\end{lemma}
\begin{proof}
  The first claim of the lemma follows immediately from Definition 3.5 and from Lemma 3.8 and Lemma 3.9 in \cite{herlihy2006dynamic}.
  The second claim follows the first claim of the lemma and the triangle inequality. 
\end{proof}
\hide{
The following simple corollary will also be useful in different places
of the analysis.
\begin{corollary}\label{cor:timesep}
  For any two requests $r=(v,t)$ and $r'=(v',t')$, if
  $t'-t > d_T(v,v')$, request $r$ is ordered before $r'$ in every
  \arrow\ execution.
\end{corollary}
\begin{proof}
  If $r=(v,t)$ is the dummy request, the statement is trivially
  true. Otherwise, assume that $r'$ is ordered before $r$. Let
  $r''=(v'',t'')$ be the immediate predecessor of request $r'$. \Cref{le:timewindow} then implies that
  $t'+d_T(v',v'')\leq t+d_T(v,v'')$ and thus
  $t'-t\leq d_T(v,v'')-d_t(v',v'')\leq d_T(v,v')$ by the triangle
  inequality.
\end{proof}
} Before delving into the details of the analysis, we give a short
outline. In the first step in \Cref{sec:partition}, we study the
ordering generated by \arrow\ in more detail and show that it implies
a hierarchical partition of the requests $R$ in a natural way. To
simplify the next \Cref{sec:split} transforms the given HST $T$
into a new tree such that inside each subtree, if ordering the request
by time of occurrence, the gap between the times of consecutive
requests cannot be too large (whenever such a gap is too large, we
split the corresponding subtree into two trees). \Cref{sec:mst}
then shows that the optimal offline cost can be characterized by the
total Manhattan cost of a spanning tree that respects the hierarchical
structure of the HST $T$ in a best given way. Finally, in \Cref{sec:arrowanalysis}, we give a general framework to compare the
queueing cost of an online distributed algorithm on an HST $T$ to the optimal
offline cost on $T$ and we apply this method to synchronous \arrow\
executions. In \Cref{sec:asynchmodel}, we show that the same
framework can also be applied to general asynchronous \arrow\
executions. 

\subsection{Characterizing \boldmath$\arrow$ By A Hierarchical Partition of $R$}
\label{sec:partition}

We hierarchically partition the requests $R$ according to the \arrow\
queueing order and the hierarchical structure of the HST $T$. On each
level $\ell$ of $T$, we partition the requests into blocks, where a
block of requests is a maximal set of requests that are ordered
consecutively by \arrow\ inside some level-$\ell$ subtree of $T$. In
the following, for non-negative integers $s$ and $t$, we use the
abbreviations $[s]:=\set{0,\dots,s-1}$ and
$[s,t]:=\set{s,\dots,t}$. Formally, instead of partitioning the set of
requests $R$ directly, we partition the set of indexes $[|R|]$. Recall
that the requests in $R$ are indexed consecutively according to the
queueing order of \arrow.

\begin{figure}[t]
    \centering
    \begin{subfigure}{0.45\textwidth}
      \centering
      \includegraphics[width=0.9\textwidth, height=3cm]{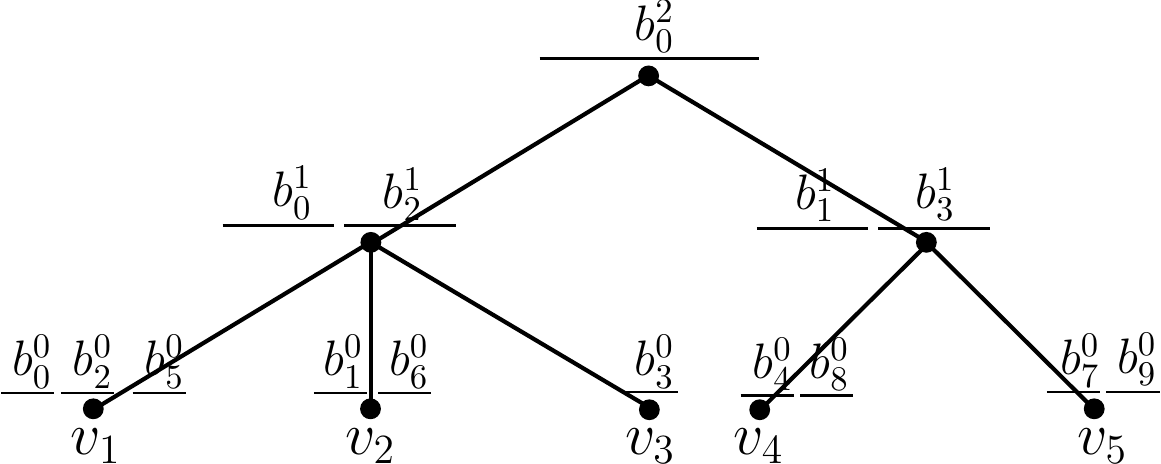}
      \caption{Blocks within the same subtree}
      \label{fig:neighborblocks}
    \end{subfigure}
    \hfill
    \begin{subfigure}{0.45\textwidth}
      \centering
      \includegraphics[width=0.9\textwidth,height=3cm]{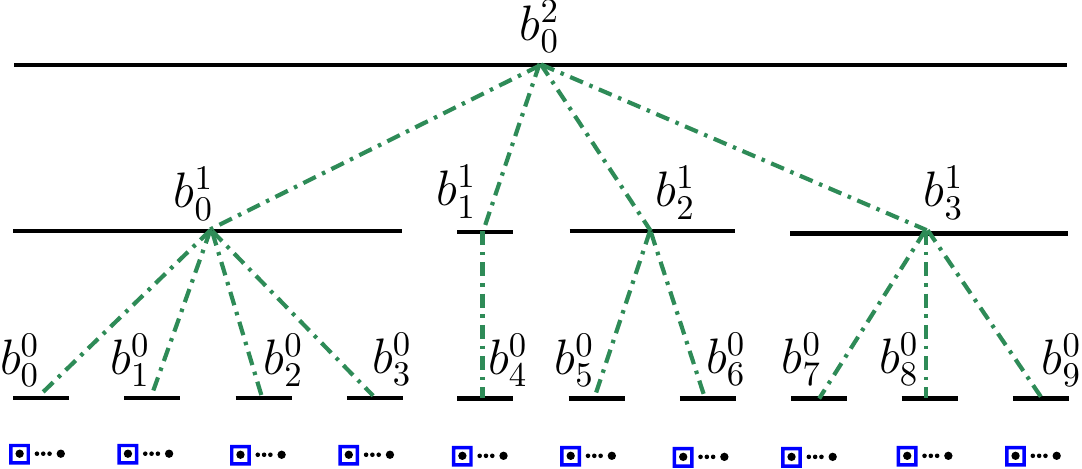}
      \caption{Tree induced by the block hierarchy}
      \label{fig:level-wise}
    \end{subfigure}
    \caption{The partition of $R$. (a) An HST with height $2$ and $5$
      leaves. The leaves issue requests at different times. The issued
      requests by nodes $v_1$, $v_2$, and $v_3$ are partitioned into
      the blocks $b^1_0$ and $b^1_2$ on level $1$. These two blocks
      are called neighbor blocks at a subtree rooted at height
      $1$. (b) The corresponding $4$ level-wise partition based on
      \arrow's order that forms a parent-child relation between the
      blocks on different levels. Blue boxes include the requests that
      are ordered first by \arrow\ among all requests in blocks
      $b^0_i$ for all $i \in [0,9]$.}
    \label{fig:firsttypepartition}
\end{figure}

\begin{definition}[\textbf{Hierarchical Block Partition}]\label{def:blockpartition}
  For each level $\lev \in [0,h]$, we partition $[|R|]$ into $n(\lev)$
  blocks $\set{b^{\lev}_0, b^{\lev}_1, \cdots , b^{\lev}_{n(\lev)-1}}$
  such that \vspace*{-3mm}
  \begin{enumerate}
  \item each block is a consecutive set of integers (i.e., a
    consecutively ordered set of requests),
  \item for every block $b_i^{\lev}$, all requests $r_p$ for $p\in
    b_i^{\lev}$ are in the same level-$\lev$ subtree of $T$, and
  \item for all $i,j\in [n(\lev)]$ and all $p\in b_i^{\lev}$ and $q\in
    b_j^{\lev}$, $i<j\ \Longrightarrow\ p<q$. 
  \end{enumerate}
  \smallskip
  For each block $b$, we further define
  the first request of $b$ to be the one that has minimum index
  in $b$.
\end{definition}

\noindent Note that for each level $\lev$ and for the first block of this
level, the first request of the block has index $0$. The block partition defined in
\Cref{def:blockpartition} is illustrated in \Cref{fig:firsttypepartition}.
\Cref{fig:neighborblocks} shows
the blocks within the HST structure, whereas \Cref{fig:level-wise} shows the hierarchical partition induced by the
blocks. To simplify the presentation of our analysis, we also define a
level $-1$ block $b^{-1}_i$ for each individual request $r_i$. Note
that we have $n(-1) =|R|$. The following definition allows to
navigate through the block hierarchy.

\begin{definition}[\textbf{Children Blocks}]\label{de:parentchild}
  The set of children blocks of a block $b_i^{\lev}$ on a level
  $\lev\in [0,h]$ is defined as
  $\child(b^{\lev}_i):=\set{b^{\lev-1}_j : b^{\lev-1}_j \subseteq
    b^{\lev}_i}$.
  Block $b_i^{\lev}$ is called the parent block of each of the blocks in
  $\child(b^{\lev}_i)$.
\end{definition}
In \Cref{fig:level-wise}, block $b^1_2$ is the parent block of its children blocks $b^0_5$ and $b^0_6$. Block $b^1_1$
has only one child block $b^0_4$ and thus $b^1_1=b^0_4$. 

The blocks
$\set{b^{\lev}_0, b^{\lev}_1, \cdots , b^{\lev}_{n(\lev)-1}}$ of level
$\lev$ belong to the subtrees rooted at height $\lev$ of the HST
$T$. Note that by the definition of the block partition, no two
consecutive blocks at the same level $\lev$ belong to the same level-$\lev$ subtree of $T$. The next
definition specifies notation to argue about blocks of the same
subtree of $T$.

\begin{definition}[\textbf{Blocks of Same Subtree}]\label{de:differentblockssamesubtree}
  If two blocks $b_i^{\lev}$ and $b_j^{\lev}$ belong to the same
  level-$\lev$ subtree of $T$, this is denoted by $\widehat{b^{\lev}_i
    b^{\lev}_j}$. Moreover, $|\widehat{b^{\lev}_i
    b^{\lev}_j}|:=\big|\big\{w : i<w<j\ \land\ \widehat{b^{\lev}_i
      b^{\lev}_w}\text{ holds}\big\}\big|$. Two blocks $b_i^{\lev}$ and
  $b_j^{\lev}$ are called \emph{neighbor blocks} if  $\widehat{b^{\lev}_i
    b^{\lev}_j}$ and $|\widehat{b^{\lev}_i b^{\lev}_j}|=0$.
\end{definition}

In \Cref{fig:neighborblocks}, blocks $b^0_0$, $b^0_2$, and
$b^0_5$ are within the same subtree rooted at node $v_1$.  Blocks
$b^0_0$ and $b^0_5$ are not neighbor blocks, however blocks $b^0_0$
and $b^0_2$, as well as blocks $b^0_2$ and $b^0_5$ are neighbor
blocks. The next lemma lists a number of simple properties of the
block partition.

\begin{lemma}\label{lemma:blockproperties}
  The block partition of \Cref{def:blockpartition} satisfies the following properties:
  \vspace*{-2mm}
  \begin{enumerate}
  \item For every block $b_i^{\lev}$ and for all $p,q\in b_i^{\lev}$,
    we have $d_T(v_p,v_q)\leq \treelength(\lev)$.
  \item For each level $\lev$ and all level-$\lev$ blocks $b_i^{\lev}$
    and $b_j^{\lev}$, if $\widehat{b^{\lev}_i b^{\lev}_j}$ holds, for
    any $p\in b_i^{\lev}$ and $q\in b_j^{\lev}$, we have
    $d_T(v_p,v_q)\leq\treelength(\lev)$.
  \item For each level $\lev$ and all level-$\lev$ blocks $b_i^{\lev}$
    and $b_j^{\lev}$, if $\widehat{b^{\lev}_i b^{\lev}_j}$ does
    \emph{not} hold, for all $p\in b_i^{\lev}$ and $q\in b_j^{\lev}$,
    we have $d_T(v_p,v_q)\geq\treelength(\lev+1)$.
  \item Assume $\lev<h$ and consider two blocks $b_i^{\lev}$ and
    $b_j^{\lev}$ that have a common parent block $b_w^{\lev+1}$, but
    for which $\widehat{b^{\lev}_i b^{\lev}_j}$ does not hold. Then, for all
    $p\in b_i^{\lev}$ and $q\in b_j^{\lev}$, we have
    $d_T(v_p,v_q)=\treelength(\lev+1)$.
  \end{enumerate}
\end{lemma}
\begin{proof}
  Recall that the distance between two leaves $u,v$ of the HST $T$ is
  equal to $\treelength(\lev)$ if the least common ancestor of $u$ and
  $v$ is on level $\lev$. The first claim then holds because all
  requests in a block $b_i^{\lev}$ at level $\lev$ are issued at nodes
  in the same level-$\lev$ subtree of $T$ and therefore the least
  common ancestor of any two of them is on level at most $\lev$. The
  second claim holds for a similar reason. If
  $\widehat{b^{\lev}_i b^{\lev}_j}$ holds for two blocks $b_i^{\lev}$
  and $b_j^{\lev}$, both blocks consist of requests in the same
  level-$\lev$ subtree of $T$. For the third claim, note that when
  $\widehat{b^{\lev}_i b^{\lev}_j}$ does not hold for two blocks
  $b_i^{\lev}$ and $b_j^{\lev}$, the two blocks do not belong to the
  same subtree at level $\ell$. Therefore for any two requests
  $p\in b_i^{\lev}$ and $q\in b_j^{\lev}$, the least common ancestor
  has to be on level at least $\ell+1$ and thus the distance
  $d_T(v_p,v_q)\geq \treelength(\lev+1)$. Finally, the fourth claim
  holds by combining the second claim (applied to block $b_w^{\lev+1}$
  on level $\lev+1$) and the third claim.
\end{proof}
\hide{
\begin{proof}
  Recall that the distance between two leaves $u,v$ of the HST $T$ is
  equal to $\treelength(\lev)$ if the least common ancestor of $u$ and
  $v$ is on level $\lev$. The first claim then holds because all
  requests in a block $b_i^{\lev}$ at level $\lev$ are issued at nodes
  in the same level-$\lev$ subtree of $T$ and therefore the least
  common ancestor of any two of them is on level at most $\lev$. The
  second claim holds for a similar reason. If
  $\widehat{b^{\lev}_i b^{\lev}_j}$ holds for two blocks $b_i^{\lev}$
  and $b_j^{\lev}$, both blocks consist of requests in the same
  level-$\lev$ subtree of $T$. For the third claim, note that when
  $\widehat{b^{\lev}_i b^{\lev}_j}$ does not hold for two blocks
  $b_i^{\lev}$ and $b_j^{\lev}$, the two blocks do not belong to the
  same subtree at level $\ell$. Therefore for any two requests
  $p\in b_i^{\lev}$ and $q\in b_j^{\lev}$, the least common ancestor
  has to be on level at least $\ell+1$ and thus the distance
  $d_T(v_p,v_q)\geq \treelength(\lev+1)$. Finally, the fourth claim
  holds by combining the second claim (applied to block $b_w^{\lev+1}$
  on level $\lev+1$) and the third claim.
\end{proof}
}

We have seen that in a synchronous \arrow\ execution, the latency cost
for ordering request $r_{i+1}$ as the successor of $r_i$ is exactly
the distance $d_T(v_i,v_{i+1})$ between the nodes of the two
requests. The total cost of \arrow\ therefore directly follows from
the structure of the block partition.

\begin{lemma}\label{lemma:blockarrowcost}
  The total cost of a synchronous \arrow\ execution on the HST $T$
  with corresponding hierarchical block partition is given by
  \[
  \cost_{\ARW}(\pi_{\ARW}) = \sum_{\lev=0}^{h-1}\big(n(\lev) - n(\lev+1)\big)\cdot\treelength(\lev+1). 
  \]
\end{lemma}
\begin{proof}
  It follows from claim 4 of \Cref{lemma:blockproperties} that
  for any two requests $r$ and $r'$, $d_T(r,r')=\delta(\lev+1)$ for
  the smallest $\lev$ for which $r$ and $r'$ are in the same
  level-$\lev$ block. The block partition implies that for every level
  $\ell$, there are $n(\lev)-1$ consecutive requests $r_i$ and
  $r_{i+1}$ which are in different level-$\lev$ blocks. For every
  $\lev\in\{0,\dots,h-1\}$, the number of consecutive request pairs at
  distance at least $\treelength(\lev+1)$ is therefore equal to
  $n(\ell)-1$. The claim of the lemma now follows because
  $\cost_{\ARW}(\pi_{\ARW}) = \sum_{i=1}^{|R|-1} d_T(v_{i-1},v_i)$.
\end{proof}
\hide{
\begin{proof}
  It follows from claim 4 of \Cref{lemma:blockproperties} that
  for any two requests $r$ and $r'$, $d_T(r,r')=\delta(\lev+1)$ for
  the smallest $\lev$ for which $r$ and $r'$ are in the same
  level-$\lev$ block. The block partition implies that for every level
  $\ell$, there are $n(\lev)-1$ consecutive requests $r_i$ and
  $r_{i+1}$ which are in different level-$\lev$ blocks. For every
  $\lev\in\{0,\dots,h-1\}$, the number of consecutive request pairs at
  distance at least $\treelength(\lev+1)$ is therefore equal to
  $n(\ell)-1$. The claim of the lemma now follows because
  $\cost_{\ARW}(\pi_{\ARW}) = \sum_{i=1}^{|R|-2} d_T(v_{i-1},v_i)$.
\end{proof}
}

\subsection{HST Conversion}
\label{sec:split}

In this section, a recursive (top-down) splitting procedure is
provided so that the original HST is converted into a new HST with
better properties. The conversion does not change the total cost of
ordering the requests by \arrow\ (in fact, it does not change the
block partition). Further, the total \manhattan \ cost of optimal offline algorithm's
order asymptotically remains unchanged as well. We describe how the
splitting procedure works and we then argue its properties.

\para{Splitting Procedure:}
We describe the splitting procedure as it is applied to a subtree $T'$
that is rooted at a given level $\lev\in\set{0,\dots,h}$ of $T$. If
$\lev=0$, the tree $T'$ is returned unchanged. Otherwise ($\lev\geq 1$), 
we go through all level-$(\lev-1)$ subtrees $T''$ of $T'$. As long as
the tree $T''$ has two neighbor blocks $b_i^{\lev-1}$ and
$b_j^{\lev-1}$ (for $i<j$) for which the following condition
\eqref{eq:split} is true, the subtree $T''$ is split into two separate
subtrees $T_1''$ and $T_2''$ of $T'$.
\begin{equation}\label{eq:split}
	t_{\min}(b^{\lev-1}_j)-t_{\max}(b^{\lev-1}_i) \geq \treelength(\lev).
\end{equation}
The splitting of $T''$ into $T_1''$ and $T_2''$ works as follows. The
topology of $T_1''$ and $T_2''$ is identical to the topology of
$T''$. Each request $r=(v,t)$ that is issued at some node $v$ of $T''$
is either placed on the isomorphic copy of $v$ in $T_1''$ or in
$T_2''$. All requests $r$ in blocks $b_x^{\lev-1}$ of $T''$ for $x\leq i$ are
placed in tree $T_1''$ and all request in blocks $b_y^{\lev-1}$ of $T''$ for
$y\geq j$ are placed in tree $T_2''$. We perform such splittings for
trees $T'$ of level $\ell$ as long as there are subtrees of $T'$ on level
$\ell-1$ with neighbor blocks that satisfy Condition
\eqref{eq:split}. As soon as no such neighbor blocks exist, the
procedure is applied recursively to all trees $T''$ at level $\lev-1$
(including the new subtrees). The whole conversion is started by
applying the procedure to the complete HST $T$.

\begin{lemma}\label{le:arrowpreservingbysplit}
  The above splitting procedure does not change the hierarchical block
  partition and it thus also preserves \arrow's queueing order
  $\pi_{\ARW}$ and its total cost $\cost_{\ARW}(\pi_{\ARW})$.
\end{lemma}
\begin{proof}
  We prove that a single splitting step does not change the block
  partition or the \arrow\ cost. The lemma then follows by induction
  on the number of splits in the above procedure.  Assume that we are
  working on tree $T'$ on level $\lev$ and that we are splitting
  subtree $T''$ of $T'$ into $T_1''$ and $T_2''$ as a result of two
  neighbor blocks $b_i^{\lev-1}$ and $b_j^{\lev-1}$ satisfying
  Condition \eqref{eq:split}.

  We first show that w.r.t.\ \arrow's ordering $\pi_{\ARW}$ before the
  splitting step, the block partition remains the same. W.r.t.\ the
  ordering $\pi_{\ARW}$, the block partition can only change if some
  block of level $\lev' \leq \lev-1$ at a subtree of $T''$ is split
  into two blocks. Note that any subtree $\tau$ of $T$ that is rooted
  at some node $v$ outside $T''$ either does not contain any node of
  $T''$ or it contains the whole subtree $T''$. In both cases, the
  request set of $\tau$ does not change and w.r.t.\ ordering
  $\pi_{\ARW}$ therefore also their blocks on the level of node $v$
  remain the same. Because the blocks at some level $\lev'<\lev-1$ of
  tree $T''$ are a refinement of the blocks on level $\lev-1$, if some
  block of some level $\lev'\leq\lev-1$ at a subtree of $T''$ is
  split, there is also a level-$(\lev-1)$ block of tree $T''$ is split
  into two blocks. However this cannot happen because the splitting
  procedure moves each level-($\lev-1$) block of $T''$ either
  completely to $T_1''$ or to $T_2''$. Hence, w.r.t.\ the ordering
  $\pi_{\ARW}$ before the splitting, the block partition remains the
  same.

  We next show that this implies that for all pairs of requests
  $(r_i,r_{i+1})$ ordered consecutively by \arrow, the tree distance
  $d_T(v_i,v_{i+1})$ remains the same. If it does not remain the same,
  it means that $v_i$ and $v_{i+1}$ are both within $T''$ and thus
  before the split $d_T(v_i,v_{i+1})\leq \treelength(\lev-1)$ (their
  least common ancestor is some node in $T''$). Hence, $r_i$ and
  $r_{i+1}$ are in the same block on level $\lev-1$. To see this,
  recall that the blocks of level $\lev-1$ of $T''$ are the maximal
  set of requests inside tree $T''$ that are ordered consecutively by
  \arrow. Because $r_i$ and $r_{i+1}$ are ordered consecutively, they
  therefore have to be in the same level $\lev-1$ block of
  $T''$. After the split, we then have
  $d_T(v_i,v_{i+1})= \treelength(\lev)$ and thus $r_i$ and $r_{i+1}$
  cannot be in the same block at level $\lev'$ any more. As the
  splitting does not change the block partition (w.r.t.\ the original
  ordering $\pi_{\ARW}$), this cannot happen. Hence, we have that for
  every $i\in\set{0,\dots,|R|-2}$, $d_T(v_i,v_{i+1})$ remains
  unchanged. All other distances can only increase. Hence, even after
  the split, for every $i\in\set{0,\dots,|R|-2}$, request $r_{i+1}$
  still minimizes $t + d_T(v,v_{i})$ among all non-ordered requests
  $r=(v,t)$. \Cref{le:timewindow} therefore implies that
  $\pi_{\ARW}$ is still a valid \arrow\ ordering. Because the block
  partition remains the same, \Cref{lemma:blockarrowcost} also
  immediately implies that $\cost_{\ARW}(\pi_{\ARW})$ remains
  unchanged.  Because when splitting tree $T''$, every
  level-($\lev-1$) block of $T''$ either completely goes to tree
  $T_1''$ or to tree $T_2''$, the splitting does not divide any
  block. Hence, if we assume that the queueing order $\pi_{\ARW}$ is
  preserved, also the block partition is preserved.
\end{proof}
The next lemma shows that if a tree $T''$ is split into two trees
$T_1''$ and $T_2''$ such that all requests in $T_1''$ are ordered
before all requests in $T_2''$, there is a significant time of occurrence gap
between the requests ending up in subtrees $T_1''$ and $T_2''$.

\begin{lemma}\label{lemma:timeseparation}
  Assume that we are performing a single splitting. Further, assume
  that we are working on a tree $T'$ on level $\lev$ and that we are
  splitting a subtree $T''$ of $T'$ into $T_1''$ and $T_2''$ such that
  $T_1''$ obtains the blocks that are scheduled first by \arrow. If
  $R_1$ and $R_2$ are the request sets of $T_1''$ and $T_2''$,
  respectively, we have
  $t_{\min}(R_2)-t_{\max}(R_1)\geq
  \treelength(\lev)-\treelength(\lev-1)$.
\end{lemma}
\begin{proof}
  Assume that the split of the tree $T''$ is caused by two neighbor
  blocks $b_i^{\lev-1}$ and $b_j^{\lev-1}$ satisfying Condition
  \eqref{eq:split}. We first show that
  $t_{\min}(R_2) = t_{\min}(b_j^{\lev-1})$. To see this, we generally
  show that for any subset of blocks $b_{i_1}^{x},b_{i_2}^x,\dots$ of
  some tree $\bar{T}$ rooted at level $x$, if $b_{i_1}^x$ is the first
  of these blocks ordered by \arrow, then the first request ordered in
  $b_{i_1}^x$ has the smallest time of occurrence among all requests
  in blocks $b_{i_1}^x,b_{i_2}^x, \dots$. To see this, note that
  whenever \arrow\ enters a level-$x$ block $b_i^x$ of tree $\bar{T}$,
  the predecessor request $r$ is at a node $v$ outside tree
  $\bar{T}$. As a consequence, all leaf nodes in $u\in \bar{T}$ and
  thus all requests in $\bar{T}$ are at the same distance from $v$ in
  the HST $T$. Therefore \Cref{le:timewindow} implies that the
  successor of $r$ is a request with minimum time of occurrence.

  It remains to show that
  \begin{equation}\label{eq:maxtimeR1}
    t_{\max}(R_1) \leq t_{\max}(b_i^{\lev-1})+\treelength(\lev-1).    
  \end{equation}
  Assume that $r_p=(v_p,t_p)$ is a request from $R_1$ with
  $t_p=t_{\max}(R_1)$. Further, assume that $r_q=(v_q,t_q)$ is the
  last request ordered by \arrow\ among the requests in $R_1$. Note
  that request $r_q$ needs to be inside block $b_i^{\lev-1}$ because
  that is the last level-($\lev-1$) block that is assigned to tree
  $T_1''$. Hence, we clearly have $t_q\leq
  t_{\max}(b_i^{\lev-1})$. Therefore, if $r_p=r_q$
  \eqref{eq:maxtimeR1} clearly holds. We can therefore assume that
  $r_p$ is ordered before $r_q$ by \arrow. Consider the predecessor
  $r_{p-1}$ of request $r_p$. From the second part of \Cref{le:timewindow}, we have
  \begin{equation}\label{eq:tptqdifference}
    t_p - t_q \leq d_T(v_p,v_q).
  \end{equation}
  Since both $r_p$ and $r_q$ are in $T''$ then $d_T(v_p,v_q) \leq \treelength(\lev-1)$
  thus \eqref{eq:maxtimeR1} holds.
\end{proof}
It remains to show that the splitting also does not affect the optimal
offline cost in a significant way. The following lemma shows that the
Manhattan cost $\mtncost_{\mtn}(r,r')$ for any two requests $r$ and
$r'$ can increase by at most a factor $3$. Hence, also the total
Manhattan cost of an optimal ordering cannot increase by more than a
factor $3$.

\begin{lemma}\label{le:split}
  For any two requests $r$ and $r'$, the splitting procedure does not
  increase the Manhattan cost $\mtncost_{\mtn}(r,r')$ by more than a
  factor $3$.  
\end{lemma}
\begin{proof}
  We prove that a) by every single splitting, the Manhattan cost
  $\mtncost_{\mtn}(r,r')$ can at most increase by a factor of $3$ and
  b) the Manhattan cost $\mtncost_{\mtn}(r,r')$ is affected by at most
  one splitting. Assume that $r=(v,t)$ and $r'=(v',t')$. Clearly, the
  issue times $t$ and $t'$ are not affected by the splitting. The
  Manhattan cost can therefore only change because $d_T(v,v')$
  changes. We first show that this can happen at most once. When
  working on tree $T'$ at level $\lev$, a splitting divides a subtree
  $T''$ at level $\lev-1$ into two subtrees $T_1''$ and
  $T_2''$. Hence, when working on level $\lev$, if two nodes are
  affected by the splitting their distance in $T'$ increases from at
  most $\treelength(\lev-1)$ to exactly
  $\treelength(\lev)$. Therefore, after separating two nodes $v$ and
  $v'$ because of a splitting for a tree $T'$ on level $\lev$, the two
  nodes cannot be affected by another splitting on a level
  $\lev'\geq \lev$. Claim b) now follows because we do the splitting
  in a top-down way, i.e., throughout the splitting procedure the
  levels on which we split are monotonically non-increasing.
  
  To prove claim a), let us assume that $r=(v,t)$ and $r'=(v',t')$ are
  affected by a splitting when a tree $T''$ at level $\lev-1$ is split
  into two trees $T_1''$ and $T_2''$. We have
  already seen that this implies that after the splitting, we have
  $d_T(v,v')=\treelength(\lev)$. It further follows from \Cref{lemma:timeseparation} that $|t-t'|\geq
  \treelength(\lev)-\treelength(\lev-1)>\treelength(\lev)/2$. Hence,
  before the splitting, we have $\mtncost_{\mtn}(r,r')\geq |t-t'|$ and
  after the splitting, we have $\mtncost_{\mtn}(r,r')\leq |t-t'| +
  d_T(v,v')<3 \cdot |t-t'|$.
\end{proof}
For the remainder of the analysis in this section (and also in \Cref{sec:asynchmodel}),
we assume that the HST $T$ is an HST that is
obtained after applying the splitting procedure recursively. We
therefore assume that for every level $\ell$ and every subtree $T'$ at
level $\ell$, there is no level-($\ell-1$) subtree $T''$ of $T'$ that
contains two neihghbor blocks that satisfy Condition \eqref{eq:split}.


\subsection{Lower Bounding The Optimal Manhattan Cost}
\label{sec:mst}

In this section, we construct a tree $\stt$ that spans all requests in
$R$. The tree $\stt$ has a nice hierarchical structure: For each
subtree $T'$ of $T$, the set edges of $\stt$ induced by the request set
of the subtree $T'$ forms a spanning tree of the request set of
$T'$. Apart from this useful structural property, we will show that the
total Manhattan cost of the spanning tree $\stt$ is within a constant
factor of minimum spanning tree (MST) of the request set $R$ w.r.t.\
the Manhattan cost. We have seen that on condensed request sets, the
optimal TSP path of the request set w.r.t.\ the Manhattan cost is
within a constant factor of the optimal offline queueing cost. Note
that because any TSP path is also a spanning tree, this implies that
the total Manhattan cost of the MST and thus also the total Manhattan
cost of the tree $\stt$ are lower bounding the optimal offline queueing
cost within a constant multiplicative factor.

Throughout this section, for convenience, we add one more level to the
HST $T$. Instead of placing the requests at the leaves on level $0$,
we assume that each level $0$ node $v$ has a child node on level $-1$
for each of the requests issued at node $v$. Hence, the new leaf nodes
are on level $-1$ and each leaf node receives exactly one
request.\footnote{Note that subtrees of $T$ that do not have any
  queueing requests can be ignored and therefore, we can w.l.o.g.\
  assume that every leaf node issues some queueing request.} The
distance between a level $-1$ node and its parent on level $0$ is set
to be $0$.

\para{Spanning Tree Construction:}
The spanning tree $\stt$ is constructed greedily in a bottom-up
fashion. For each subtree $T'$ of $T$, we recursively define a tree
$\stt(T')$ as follows. For the leaf nodes on level $-1$, the tree
consists of the single request placed at the node. For a tree $T'$
rooted at a node $v$ on level $\ell\geq 0$, the tree $\stt(T')$
consists of the recursively constructed trees
$\stt(T_1''), \stt(T_2''), \dots$ of the subtrees
$T_1'', T_2'', \dots$ of $T''$ and of edges connecting the trees
$\stt(T_1''), \stt(T_2''), \dots$ to a spanning tree of the set of
request issued at leaves of tree $T'$. The edges for connecting the
trees $\stt(T_1''), \stt(T_2''), \dots$ are chosen so that they have
minimum total Manhattan cost. That is, to connect the trees
$\stt(T_1''), \stt(T_2''), \dots$, we compute an MST of the graph we
get if each of the trees $\stt(T_i'')$ is contracted to a single
node. We can therefore for example choose the edges to connect the
trees $\stt(T_1''), \stt(T_2''), \dots$ in e greedy way: Always add
the lightest (w.r.t.\ Manhattan cost) edge that does not close a cycle
with the already existing edges, including the edges of the trees
$\stt(T_1''), \stt(T_2''), \dots$.

\para{MST Approximation:} In the following, it is shown that the total
Manhattan cost of the tree $\stt=\stt(T)$ is within a constant factor
of the cost of an MST w.r.t.\ the Manhattan cost. Where convenient, we
identify a tree $\tau$ with its set of edges, i.e., we also use $\stt$
to denote the set of edges of the tree $\stt$. Further, the cost of an
edge $e=\set{r,r'}$ is the Manhattan cost $\mtncost_{\mtn}(r,r')$. We
also slightly abuse notation and use $\mtncost_{\mtn}(e)$ to denote
this cost. The proof applies a general MST approximation result that
appears in \Cref{thm:MSTapprox} in \Cref{sec:MSTapprox}. Together with the following technical lemma,
\Cref{thm:MSTapprox} directly implies that the total Manhattan
cost of $\stt$ is within a factor $4$ of the MST Manhattan cost. For a
subtree $T'$ of $T$, we use $R(T')$ to denote the subset of the
requests $R$ that are issued at nodes of $T'$.

\begin{figure}[t]
  \centering
  \includegraphics[width=7.5cm, height=4cm]{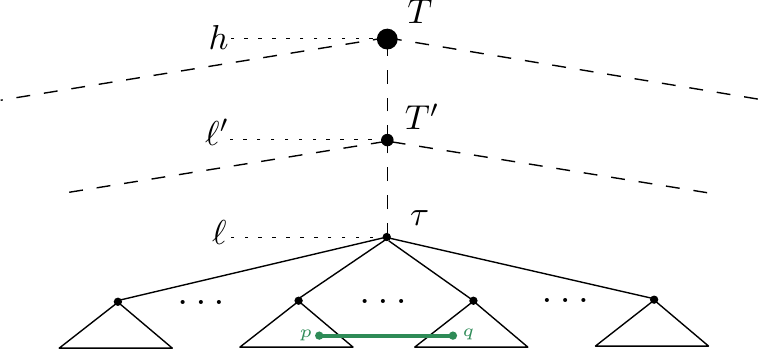}
  \caption{The HST $T$ and the edge $\set{r_p,r_q}$. The subtree
    $T'$ is the highest subtree that includes $r_p$ and $r_q$ while
    $|t_p-t_q| > 3 \cdot \treelength(\lev')$ where $T'$ is rooted at
    level $\lev' \geq \lev$.}
  \label{fig:mst}
\end{figure}

\begin{lemma}\label{le:edgeapprox}
  Consider the constructed spanning tree $\stt$ and consider an
  arbitrary edge $e$ of $\stt$. Let $\stt_1$ and $\stt_2$ the two
  subtrees that result when removing edge $e$ from $\stt$. Further,
  assume $e^*$ be an edge that connects the two subtrees $\stt_1$ and
  $\stt_2$ and that has minimum Manhattan cost among all such
  edges. We then have
  $\mtncost_{\mtn}(e) \leq 4 \cdot \mtncost_{\mtn}(e^*)$.
\end{lemma}
\begin{proof}
  Assume that the edge $e=\set{r_p,r_q} \in \stt(\tau)$ is an edge
  that connects two subtrees of a subtree $\tau$ of $T$ that is rooted
  at some level $\lev\in [0,h]$. Further, let $V_{\stt_1}$ and
  $V_{\stt_2}$ be the node sets of the two subtrees of $\stt_1$ and
  $\stt_2$.

  Let us first assume that $|t_p-t_q| \leq 3 \cdot \treelength(\lev)$.
  All edges including $e^*$ from the metric $(R,\mtncost_{\mtn})$
  that cross the cut $(V_{\stt_1},V_{\stt_2})$ have length at least
  $\treelength(\lev)$ since $d_T(v_w,v_z) \geq \treelength(\lev)$ for
  all $r_w \in V_{\stt_1}$ and $r_z \in V_{\stt_2}$. Since
  $d_T(v_p,v_q)=\treelength(\lev)$, we then have
  $\mc(e) \leq 4 \cdot \treelength(\lev)$. Hence, the
  claim of the lemma holds.

  Let us therefore assume that
  $|t_p-t_q| > 3 \cdot \treelength(\lev)$. Let $\lev' \in [\lev,h]$ be
  the largest level for which $|t_p-t_q| > 3 \cdot \treelength(\lev')$
  and let $T'$ be the subtree of $T$ that is rooted on level $\lev'$
  and that contains both requests $r_p$ and $r_q$ (see \Cref{fig:mst}). Note that this implies that  
  \begin{equation}\label{eq:upperboundone}
    |t_p-t_q| \leq 3\cdot\treelength(\lev'+1)\quad\text{and thus}\quad
    \mc(e) \leq 3\cdot\treelength(\lev'+1) + \treelength(\lev).
  \end{equation}
  
  We can partition each of the sets $V_{\stt_1}$ and $V_{\stt_2}$ into
  two sets where one of the sets in each case includes the requests in
  the subtree $T'$ and the other set includes the requests outside
  subtree $T'$ (see \Cref{fig:edgeapprox}). The edge $e$
  obviously connects the two components $V_{\stt_1} \cap R(T')$ and
  $V_{\stt_2} \cap R(T')$ since $r_p$ and $r_q$ are both in
  $R(T')$. If the edge $e$ is removed then edge $e^*$ is an edge
  connecting one of the two components $V_{\stt_1} \cap R(T')$ and
  $V_{\stt_1} \setminus R(T')$ in $V_{\stt_1}$ to one of the two
  components $V_{\stt_2} \cap R(T')$ and
  $V_{\stt_2} \setminus R(T')$ in $V_{\stt_2}$. The four different
  types of such edges are shown by the dashed edges in \Cref{fig:edgeapprox}.
  \begin{figure}[t]
    \centering
    \includegraphics[width=9cm, height=4cm]{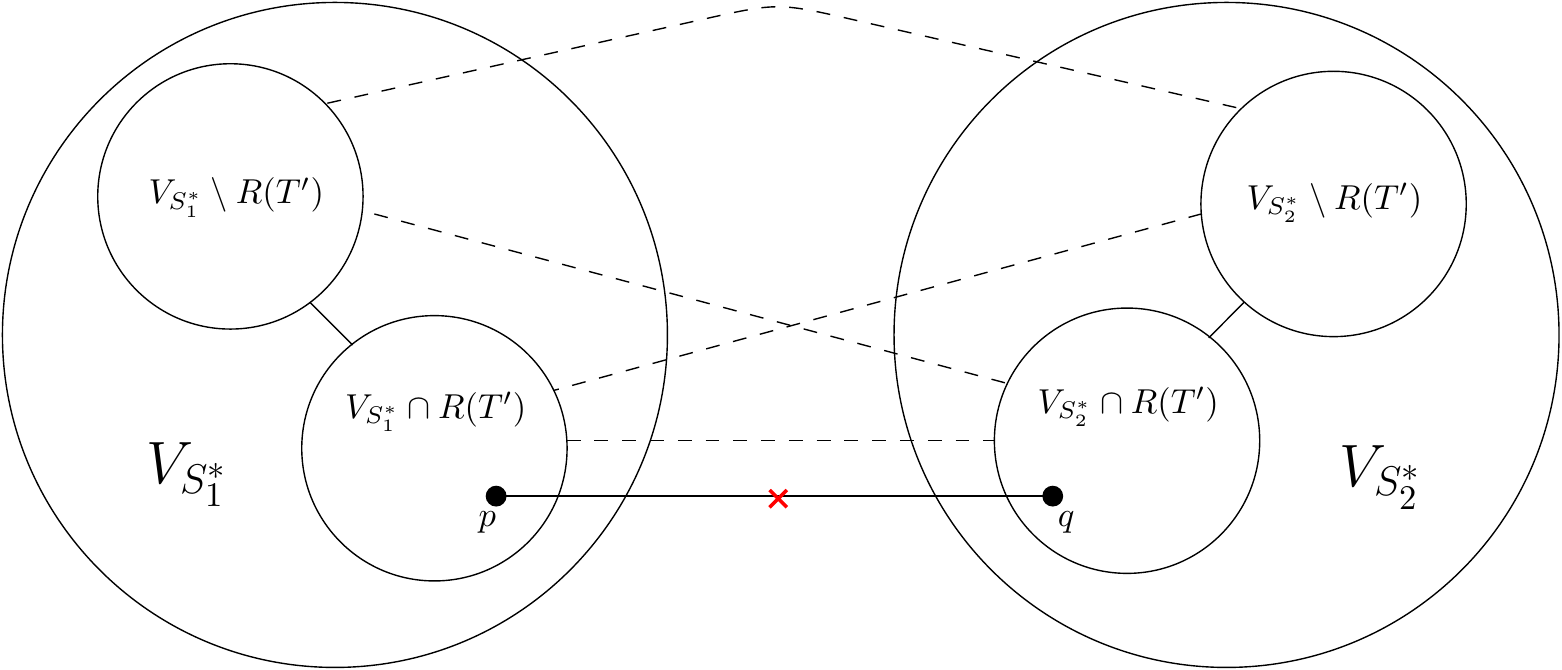}
    \caption{The spanning tree $\stt$ when there is a subtree $T'$
      that is rooted at height $\lev' \in [\lev,h]$ and is the highest
      subtree where $|t_p-t_q| > 3 \cdot \treelength(\lev')$. If the
      edge $\set{p,q}$ is removed then the edge $e^*$ could be one of the
      dashed edges.}
    \label{fig:edgeapprox}
  \end{figure}

  Any edge that connects the two components
  $V_{\stt_1}\setminus R(T')$ and $V_{\stt_2} \cap R(T')$ has length
  at least $\treelength(\lev'+1)$ since
  $d_T(v_w,v_z) \geq \treelength(\lev'+1)$ for all
  $r_w \in V_{\stt_1}\setminus R(T')$ and
  $r_z \in V_{\stt_2} \cap R(T')$. By symmetry, the same also holds
  for the edges that connect the two components
  $V_{\stt_1} \cap R(T')$ and $V_{\stt_2}\setminus R(T')$. Hence, if
  $e^*$ is an edge of one of these two types, we have
  $\mc(e^*)\geq \treelength(\lev'+1)$. It then follows directly from
  \eqref{eq:upperboundone} that $\mc(e)\leq 4 \cdot \mc(e^*)$ and thus the
  claim of the lemma holds.

  Let us therefore move to the case where $e^*$ connects the two
  components $V_{\stt_1} \setminus R(T')$ and
  $V_{\stt_2}\setminus R(T')$, i.e., $e^*=\set{r_x,r_y}$ connects to
  nodes $v_x$ and $v_y$ outside tree $T'$. Recall that the tree $\stt$
  is constructed in a bottom-up way such that the subtree $\stt(T'')$
  of $\stt$ is connected for every subtree $T''$ of $T$. Hence,
  removing edge $e$ inside subtree $T'$ does not affect subtrees
  $\stt(T'')$ for trees $T''$ that do not contain $T'$. Therefore if
  two nodes $u$ and $v$ outside tree $T'$ end up on different sides of
  the cut $(V_{\stt_1},V_{\stt_2})$, the least common ancestor of
  $v_x$ and $v_y$ has to be an ancestor of $T'$ and it is thus at
  level at least $\lev'+1$. Hence, if $e^*$ connects the two
  components $V_{\stt_1} \setminus R(T')$ and
  $V_{\stt_2}\setminus R(T')$, we also have
  $\mc(e^*)\geq\treelength(\lev'+1)$ and therefore again
  \eqref{eq:upperboundone} implies the claim of the lemma.
  
  It remains to show that all edges that connect the two components
  $V_{\stt_1} \cap R(T')$ and $V_{\stt_2} \cap R(T')$ are also large
  enough. W.l.o.g., we assume that $p < q$, i.e., the request $r_p$ is
  ordered before the request $r_q$ by \arrow. Further w.l.o.g., we
  assume that the dummy request is in $V_{\stt_1}$.

  We next show that $t_q > t_p$.  If $p=0$ then the $t_q \geq t_p$
  because $t_p=0$ and because $|t_p-t_q|>3 \cdot \treelength(\lev)\geq
  0$.
  Otherwise, for the sake of contradiction, let us assume that
  $t_q \leq t_p$. By the second part of \Cref{le:timewindow} we have
  \[
  t_p-t_q \leq d_T(v_p,v_q) \leq \treelength(\lev).
  \]
  This together with our assumption $t_q \leq t_p$ contradicts the fact
  that $|t_p-t_q| > 3 \cdot \treelength(\lev)$. Therefore,
  $t_q > t_p$. 

  Recall that $e$ connects the two requests $r_p$ and $r_q$ inside
  level-$\lev$ tree $\tau$. Consider the subtree $\stt(\tau)$ of
  $\stt$ and let $\stt_1(\tau)$ and $\stt_2(\tau)$ be the two subtrees
  of $\stt(\tau)$ that are obtained when removing edge $e$ from
  $\stt(\tau)$.  By the construction of the tree \stt, the edge $e$ is
  one with minimum Manhattan cost among all edges connecting the
  requests in $\stt_1(\tau)$ and $\stt_2(\tau)$. We know that for all
  $r_w \in V_{\stt_1(\tau)}$ and $r_z \in V_{\stt_2(\tau)}$ we have
  $d(v_w,v_z)=\treelength(\lev)$.  These facts imply that
  $t_p=t_{\max}(V_{\stt_1(\tau)})$ and
  $t_q=t_{\min}(V_{\stt_2(\tau)})$.

  Now we show that there is an \arrow \ edge $(r_x,r_{x+1})$ where
  $r_x \in V_{\stt_1(\tau)}$ and $r_{x+1} \in V_{\stt_2(\tau)}$.  For
  any two neighbor blocks $b^{\lev}_i$ and $b^{\lev}_j$ at
  subtree $\tau$ and with $i<j$, we know that
  \[
  t_{\min}(b^{\lev}_j)-t_{\max}(b^{\lev}_i) < \treelength(\lev+1)
  \] 
  as otherwise because of the split condition \eqref{eq:split}, the subtree
  $\tau$ would have been split. Thus, we have
  \[
  t_{\min}(b^{\lev}_j)-t_{\max}(b^{\lev}_i) < 3 \cdot \treelength(\lev)
  \] 
  since $\treelength(\lev+1) \leq 3 \cdot \treelength(\lev)$ for
  $\alpha=2$.  Let
  $b_{i_1}^{\lev},b_{i_2}^{\lev},\dots,b_{i_s}^{\lev}$ be the
  level-$\lev$ blocks of the subtree $\tau$ and assume that
  $i_1<i_2<\cdots<i_s$.  As $t_q-t_p > 3 \cdot \treelength(\lev)$ and
  because $t_p=t_{\max}(V_{\stt_1(\tau)})$ and
  $t_q=t_{\min}(V_{\stt_2(\tau)})$, for any two neighbor blocks
  $b_{i_j}^{\lev}$ and $b_{i_{j+1}}^{\lev}$, the requests $r=(v,t)$
  from $b_{i_j}^{\lev}$ with $t=t_{\max}(b_{i_j}^{\lev})$ and the
  requests $r'=(v',t')$ from $b_{i_{j+1}}^{\lev}$ with
  $t'=t_{\min}(b_{i_{j+1}}^{\lev})$ either all have to be in in
  $V_{\stt_1(\tau)}$ or they all have to be in $V_{\stt_2(\tau)}$. We
  show that this implies that there has to be a block $b_{i_j}^{\lev}$
  at tree $\tau$ for which the first request is in $V_{\stt_1(\tau)}$
  and which contains some request from $V_{\stt_2(\tau)}$. First note
  that because of \Cref{le:timewindow} and because we assumed
  that the dummy request is in $V_{\stt_1(\tau)}$, the first request
  of $b_{i_1}^{\lev}$ is in $V_{\stt_1(\tau)}$. If all the first
  requests of blocks $b_{i_j}^{\lev}$ are in $V_{\stt_1(\tau)}$, it
  follows from the fact that $V_{\stt_2(\tau)}$ needs to be non-empty
  that there has to be a block $b_{i_j}^{\lev}$ for which the first
  request is in $V_{\stt_1(\tau)}$ and which contains some request
  from $V_{\stt_2(\tau)}$. Otherwise, assume that $b_{i_j}^{\lev}$
  (for $j\geq 2$) is the first block for which the first request is in
  $V_{\stt_2(\tau)}$. Because by \Cref{le:timewindow}, the first
  request of a block is always one with smallest issue time, the above
  observation implies that the request with the largest issue time in
  $b_{i_{j-1}}^{\lev}$ is in $V_{\stt_2(\tau)}$ and then
  $b_{i_{j-1}}^{\lev}$ there has the first request is in
  $V_{\stt_1(\tau)}$ and which contains some request from
  $V_{\stt_2(\tau)}$. In a block, where the first request is from
  $V_{\stt_1(\tau)}$ and there is some request from
  $V_{\stt_2(\tau)}$, there also have be two consecutive requests
  $r_x$ and $r_{x+1}$ (and thus an \arrow\ edge), such that
  $r_x \in V_{\stt_1(\tau)}$ and $r_{x+1} \in V_{\stt_2(\tau)}$.

  We next show that the \arrow\ edge $(r_x,r_{x+1})$ is the only such
  \arrow\ edge even with respect to the tree $T'$ containing tree
  $\tau$. Specifically, we show that for all
  $r_w \in V_{\stt_1} \cap R(T')$ and all
  $r_z \in V_{\stt_2} \cap R(T')$ we have $w \leq x$ and
  $z \geq x+1$. In other words, $r_x$ is the last request ordered in
  $V_{\stt_1} \cap R(T')$ and $r_{x+1}$ is the first request
  ordered in $V_{\stt_2} \cap R(T')$. For the sake of contradiction,
  let us assume that there is a request
  $r_w \in V_{\stt_1} \cap R(T')$ for which $w > x$ or that there is
  a request $r_z\in V_{\stt_2}\cap R(T')$ for which $z<x+1$. We first
  assume the existence of request $r_w$. Since $(x,x+1)$ is an \arrow
  \ edge, we have $w>x+1$ and using the second part of \Cref{le:timewindow} we
  get
  \[
  t_{x+1}-t_w \leq d(v_w,v_{x+1}) \leq \treelength(\lev').
  \]
  However, we know that $t_q -t_p \leq t_{x+1}-t_w$ and therefore 
  \[
  t_q -t_p \leq \treelength(\lev').
  \]
  This contradicts the fact that
  $t_q-t_p > 3 \cdot \treelength(\lev')$. Consequently, there does not
  exist any requests $r_w \in V_{\stt_1} \cap R(T')$ for which
  $w > x$. Now, let us assume that there is a request
  $r_z \in V_{\stt_2} \cap R(T')$ for which $z < x+1$. Again since
  $(x,x+1)$ is an \arrow \ edge, we have $z < x$ and using the second part of
  \Cref{le:timewindow} we get
  \[
  t_z-t_x \leq d(v_z,v_x) \leq \treelength(\lev').
  \]
  However, we know that $t_q - t_p \leq t_z - t_x$ and therefore 
  \[
  t_q - t_p \leq \treelength(\lev').
  \]
  Again, this is a contradiction to the fact that
  $t_q-t_p > 3 \cdot \treelength(\lev')$. Consequently, there does not
  exist any requests $r_z \in V_{\stt_2} \cap R(T')$ with $z < x+1$.

  Finally we show that for all $r_w \in V_{\stt_1} \cap R(T')$ and
  all $r_z \in V_{\stt_2} \cap R(T')$ the Manhattan cost
  $\mc(r_w,r_y)$ is at most $3 \cdot \treelength(\lev')$.  Using 
  the second part of \Cref{le:timewindow} we have
  \begin{equation}\label{eq:up1}
    t_{x+1}-t_z \leq d(v_z,v_{x+1}) \leq \treelength(\lev').
  \end{equation}
  We can similarly bound $t_w-t_x$. If $w=0$ we have $t_w \leq t_x$
  and otherwise, using the second part of \Cref{le:timewindow} we have
  \begin{equation}\label{eq:up2}
    t_w-t_x \leq d(v_x,v_w) \leq \treelength(\lev').
  \end{equation}
  Using \eqref{eq:up1} and \eqref{eq:up2} we then get
  \begin{equation}\label{eq:up3}
    t_{x+1}-t_x \leq t_z-t_w + 2 \cdot \treelength(\lev').
  \end{equation}
  We know that the Manhattan cost of $(r_x,r_{x+1})$ is at least the
  Manhattan cost of $(r_p,r_q)$ because $t_q - t_p \leq t_{x+1} - t_x$
  and because for all $r_f \in V_{\stt_1(\tau)}$ and
  $r_g \in V_{\stt_2(\tau)}$, we have
  $d(v_f,v_g)=\treelength(\lev)$. That is, we have
  \[
  \mc(r_p,r_q) \leq \mc(r_x,r_{x+1}). 
  \]
  Further, because for all $r_f \in V_{\stt_1} \cap R(T')$ and $r_g
  \in V_{\stt_2} \cap R(T')$, we have $d(v_f,v_g) \geq
  \treelength(\lev)$, by using \eqref{eq:up3}, we obtain
  \begin{equation}\label{eq:up4}
    \mc(r_p,r_q) \leq \mc(r_x,r_{x+1}) \leq \mc(r_w,r_z) + 2 \cdot \treelength(\lev').
  \end{equation}
  Therefore, by using the facts that
  $t_q - t_p > 3 \cdot \treelength(\lev')$ and
  $t_q - t_p \leq t_{x+1}-t_x$, and by using \eqref{eq:up3}, we get
  that $t_z-t_w \geq t_{x+1}-t_x - 2\treelength(\lev') \geq
  \treelength(\lev')$ and we thus have $\mc(r_w,r_z)\geq
  \treelength(\lev')$. By applying \eqref{eq:up4}, we thus get that
  \[
  \mc(r_p,r_q) < 3 \cdot \mc(r_w,r_z).
  \]
  Consequently, also if $e^*$  connects the two components $V_{\stt_1}
  \cap R(T')$ and $V_{\stt_2} \cap R(T')$, its Manhattan cost is
  within a factor $3$ of the Manhattan cost of $e$. Hence, the claim
  of the lemma holds.
\end{proof}
\begin{corollary}\label{co:mstapprox}
  The total Manhattan cost of the spanning tree $\stt$ is at most 4
  times the total Manhattan cost of an MST spanning all the requests.
\end{corollary}
\begin{proof}
  Follows directly from \Cref{le:edgeapprox} and \Cref{thm:MSTapprox}.
\end{proof}

\hide{
\begin{lemma}\label{le:mstapprox}
Consider the spanning tree $\stt$ that is resulted by the construction described in \Cref{sec:mstconstruction}.
The claim is that
\[
	\mtntotalcost_{\mtn}(\stt) \leq 4 \cdot \mtntotalcost_{\mtn}(\mst).
\] 
\end{lemma}
\begin{proof}
It is sufficient to define a one-to-one correspondence between the edge set of $\stt$ and the edge set of $\mst$ such that
for any two edges $e \in E_{\stt}$ and $\hat{e} \in E_{\mst}$ that are paired we can have $\length(e) \leq 4 \cdot \length(\hat{e})$.

This assignment between edges is iteratively defined where in each iteration an edge of $\stt$ is replaced with an edge
of $\mst$ until we get the MST $\mst$. Let $\stt(i)$ denote the spanning tree after $i$-th replacement where $\stt(0):=\st$.
In each iteration $i \geq 1$, we consider an arbitrary edge $e_i \in E_{\stt(i-1)} \cap E_{\stt}$. Let $\stt_1(i-1)$
and $\stt_2(i-1)$ denote the spanning trees that are resulted if $e_i$ is removed from $\stt(i-1)$. Consider the edge $\hat{e}_i$ that crosses
the cut $(V_{\stt_1(i-1)},V_{\stt_2(i-1)})$ and has minimum length among all edges in $E_{\mst}$ that cross the cut.
The edge $e_i$ is replaced with the edge $\hat{e}_i$ in iteration $i$. Further, let $e^*_i$ denote an edge that crosses the cut
$(V_{\stt_1(i-1)},V_{\stt_2(i-1)})$ and has minimum length among all edges from the metric $(R,\mtncost_{\mtn})$ that cross the cut.
We need to prove two things: first,
\[
	\length(e_i) \leq 4 \cdot \length(\hat{e}_i)
\]
for all $i\geq1$ and second the edge $\hat{e}_i$ (for all $i\geq 1$) has not been paired with any other edges $e_k$ for all $k \in [i-1]$.
Using the \Cref{le:edgeapprox}, for all $i \geq 1$ we have $\length(e_i) \leq 4 \cdot \length(e^*_i)$ and since $\length(e^*_i)
\leq \length(\hat{e}_i)$ then we have $\length(e_i) \leq 4 \cdot \length(\hat{e}_i)$.
In the first iteration, the edge $e_1 \in E_{\stt}$ is paired with $\hat{e}_1$ and therefore replaced with it. It is obvious that the edge
$\hat{e}_1$ has not been paired with any other edges in $E_{\stt}$ so far.
Now for $k \in [i]$, assume the edges $e_k$ are paired and replaced with edges $\hat{e}_k$ and the spanning tree $\stt(i)$ is resulted.
In iteration $i+1$, the edge $e_{i+1}$ is paired with $\hat{e}_{i+1}$ and therefore replaced with that.
The edge $\hat{e}_{i+1}$ has not been paired with any edges $e_k$ for $k \in [i]$ yet since all edges $\hat{e}_k$ for $k \in [i]$ are either in
$E_{\stt_1(i)}$ or in $E_{\stt_2(i)}$. 
\end{proof}
}

\section{Analysis of the Online Queueing Cost}
\label{sec:arrowanalysis}

In this section, we give a general framework to compare the queueing
cost of an online queueing algorithm on HST $T$ with the bound of the
offline queueing cost as established in \Cref{sec:treeanalysis}. At the end
of the section, we apply the method
to analyze synchronous \arrow\ executions on $T$. As in \Cref{sec:mst},
for convenience, we add one more level to the HST $T$
so that each level $0$ node $v$ has a child node on level $-1$ for
each of the requests issued at node $v$. The new leaf nodes are on
level $-1$ and each leaf node receives exactly one request.

We first state two basic locality properties of \arrow\ and possibly
other online queueing protocols. We will then show that those
properties are sufficient to prove a constant competitive ratio
compared to the optimal offline queueing cost on $T$. We define the
notion of a \emph{distance-respecting queueing order} and the notion
of \emph{distance-respecting latency cost} of a queueing algorithm.

\begin{definition}[\textbf{Distance-Respecting Order}]\label{def:distorder}
  Let $R$ be a set of requests $r_i=(v_i,t_i)$ issued at the nodes of
  a tree $T$ and let $\pi$ be permutation on $[0,|R|-1]$. The ordering
  $r_{\pi(0)},r_{\pi(1)},\dots,r_{\pi(|R|-1)}$ induced by $\pi$ is
  called \emph{distance-respecting} if whenever $\pi(i)<\pi(j)$, we
  have $t_i-t_j\leq d_T(v_i,v_j)$.
\end{definition}

\begin{definition}[\textbf{Distance-Respecting Latency Cost}]\label{def:distlatency}
  An online distributed queueing algorithm \ALG\ is said to have
  \emph{distance-respecting latency cost} if for any request set $R$ and
  any possible queueing order  $\pi_{\ALG}$ of $\ALG$, for all $1\leq
  i <j <|R|$, it holds that \vspace*{-2mm}
  \[t_{\pi_{\ALG}(i)} +
  \latency_{\ALG}(r_{\pi_{\ALG}(i),\pi_{\ALG}(i-1)}) \leq
  t_{\pi_{\ALG}(j)} + d_T(v_{\pi_{\ALG}(j)},v_{\pi_{\ALG}(i-1)}).\]
\end{definition}

\subsection{Constructing a Spanning Tree}
As the first part of the online queueing cost analysis, we construct a
new tree $\starw$ that spans all
requests in $R$. It will be shown that the total \manhattan \ cost of
$\starw$ asymptotically equals the total \manhattan \ cost of the tree
$\stt$ constructed in the previous section.

We construct a new tree $\starw$ on $R$ based on an ordering $\pi$ of
the set of requests. We assume that the ordering of the requests given
by $\pi$ is $r_{\pi(0)},r_{\pi(1)},\dots,r_{\pi(|R|-1)}$. For each
index $i$ with $i\in [0,|R|-2]$, we define the \emph{local
  successor} as
\begin{equation}\label{eq:synchronousnext}
  \nextt(i) := \min\set{j\in [i+1,|R|-1]\,:\, d_T(v_{\pi(i)},v_{\pi(j)}) =
    \min_{k\in [i+1,|R|-1]} d_T(v_{\pi(i)},v_{\pi(k)})}.
\end{equation}
Hence, among the requests ordered after $r_{\pi(i)}$ by order $\pi$,
$\nextt(i)$ is the position of a request in the order $\pi$
with minimum tree distance to $v_{\pi(i)}$ and among those, of the first
one ordered by $\pi$.  Note that this means that for all requests
$r_{\pi(k)}$ for which $i<k<\nextt(i)$, we have
$d_T(v_{\pi(i)},v_{\pi(k)})>d_T(v_{\pi(i)},v_{\pi(\nextt(i))})$ and for all requests $r_{\pi(k)}$ for
which $k\geq \nextt(i)$, we have
$d_T(v_{\pi(i)},v_{\pi(k)})\geq d_T(v_{\pi(i)},v_{\pi(\nextt(i))})$.

The spanning tree $\starw$ is constructed as follows. For every
request $r_{\pi(i)}$ for all $i\in[0,|R|-2]$, we add the edge
$\set{r_{\pi(i)},r_{\pi(\nextt(i))}}$ to the tree $\starw$. Note that $\starw$ is
indeed a spanning tree: If directing each edge from $r_{\pi(i)}$ to
$r_{\pi(\nextt(i))}$, each node has out-degree $1$ and we cannot have cycles
because $\nextt(i)>i$. The following observation shows that in
addition, $\starw$ has the same useful hierarchical structure as the
tree $\stt$ constructed in \Cref{sec:mst}.

\begin{observation}\label{ob:arrowspanningtree}
  As the tree $\stt$, also the tree $\starw$ has the property that for
  any subtree $T'$ of $T$, the subgraph of $\starw$ induced by only
  the requests at nodes in $T'$ is a connected subtree of $\starw$.
  This follows directly from the definition of the local successor
  $r_{\pi(\nextt(i))}$. Except for the last ordered request inside $T'$, the
  local successor of any other request of $T'$ is inside $T'$ (because
  the local successor is a request with minimum tree distance).\qed
\end{observation}

In light of Observation \ref{ob:arrowspanningtree}, for any subtree
$T'$ of $T$, we use $\starw(T')$ to denote the subtree of $\starw$
induced by the requests issued at nodes in $T'$.

\subsection{Bounding the Manhattan Cost of the Spanning Tree}
\label{sec:generalframeworkMTNcost}
The following lemma shows that if the spanning tree $\starw$ is
constructed by using a distance-respecting ordering $\pi$, the total
Manhattan cost of the spanning tree $\starw$ is asymptotically equal
the total Manhattan cost of $\stt$.

\begin{lemma}\label{le:arrowmstapprox}
  Let $\mtntotalcost_{\mtn}(\starw)$ and
  $\mtntotalcost_{\mtn}(\stt)$ be the total Manhattan costs of
  $\starw$ and of $\stt$. If the tree $\starw$ is constructed using a
  distance-respecting ordering $\pi$, we have
  $\mtntotalcost_{\mtn}(\starw) \leq 3 \cdot \mtntotalcost_{\mtn}(\stt)$. 
\end{lemma}
\begin{proof}
  Consider some subtree $\tau$ of $T$ that is rooted at a node on level
  $\lev \in [0,h]$. Assume that $v$ has $m$ children an that the
  subtrees of $T$ rooted at the $m$ children are
  $\tau_1,\tau_2, \ldots, \tau_m$. Using Observation \ref{ob:arrowspanningtree}, we know that
  $\starw(\tau_1),\starw(\tau_2), \ldots, \starw(\tau_m)$ are subtrees
  of $\starw(\tau)$ trees that are connected to each other with $m-1$
  edges to form the spanning tree $\starw(\tau)$. Let us call this set
  of edges $\Iarw(\tau)$.  Note that for $\lev=0$ the subtrees of
  $\tau$ are single requests at level $-1$.  Similarly, the
  construction of $\stt$ implies that the spanning tree $\stt(\tau)$
  results from connecting the spanning trees
  $\stt(\tau_1),\stt(\tau_2), \ldots, \stt(\tau_m)$ with $m-1$
  edges. Let $\I(\tau)$ denote this set of these $m-1$ edges. Recall
  that the edges in $\I(\tau)$ are chosen such that they have minimum
  total Manhattan cost among all sets of $m$ edges connecting the
  trees $\stt(\tau_1),\stt(\tau_2), \ldots, \stt(\tau_m)$. We also
  emphasize that for all $i\in[1,m]$, the trees $\starw(\tau_i)$ and
  $\stt(\tau_i)$ consist of the same set of nodes (the requests inside
  tree $\tau_i$). Let $\mtntotalcost_{\mtn}(\Iarw(\tau))$ and
  $\mtntotalcost_{\mtn}(\I(\tau))$ be the total Manhattan costs of the
  edges in $\Iarw(\tau)$ and $\I(\tau)$, respectively. To prove the
  lemma, it suffices to show that
  \begin{equation}\label{eq:arrowMC1}
    \text{$\forall$ subtree $\tau$ of $T$:}\
    \mtntotalcost_{\mtn}(\Iarw(\tau))\leq
    3\cdot\mtntotalcost_{\mtn}(\I(\tau)).
  \end{equation}
  Let $e=(r_{\pi(w)},r_{\pi(z)}) \in \Iarw(\tau)$ be an arbitrary edge of
  $\Iarw(\tau)$ and let $\starw_1(\tau)$ and $\starw_2(\tau)$ be the
  two subtrees of $\starw(\tau)$ resulting from removing $e$ from
  $\Iarw(\tau)$. Let $V_{\starw_1(\tau)}$ and $V_{\starw_2(\tau)}$ be
  the set of nodes (requests) of the trees $\starw_1(\tau)$ and
  $\starw_2(\tau)$ and assume, w.l.o.g., that $w<z$ and that
  $r_{\pi(w)} \in V_{\starw_1(\tau)}$ and $r_{\pi(z)} \in V_{\starw_2(\tau)}$. Also,
  consider an edge $e^*$ that crosses the cut
  $(V_{\starw_1(\tau)},V_{\starw_2(\tau)})$ and has minimum Manhattan
  cost among all edges in $\stt(\tau)$ that cross this cut. Note that
  because for all $i$ the trees $\starw(\tau_i)$ and $\stt(\tau_i)$
  consist of the same set of node, node $e^*$ must be from the set
  $\I(\tau)$. In order to prove \eqref{eq:arrowMC1}, it suffices to
  show that
  \begin{equation}\label{eq:arrowMC2}
    \mc(e) \leq 3 \cdot \mc(e^*).
  \end{equation}
  Inequality \eqref{eq:arrowMC1} then directly follows from
  \Cref{thm:MSTapprox}.

  From the definition of local successor, we know that
  $z=\nextt(w)$. This implies that for all requests $r_{\pi(x)}$ where
  $w < x < z$, we have $d_{T}(v_{\pi(w)},v_{\pi(x)}) > \treelength(\lev)$ since
  $d_{T}(v_{\pi(w)},v_{\pi(z)}) = \treelength(\lev)$. Therefore, all requests
  that are ordered between $r_{\pi(w)}$ and $r_{\pi(z)}$ by \arrow \ are not in
  $R(\tau)$ (i.e., in the set of requests of tree $\tau$). This means
  that all requests in $R(\tau)$ are ordered either before $r_{\pi(w)}$ or
  after $r_{\pi(z)}$ by \arrow. More precisely, the claim is that for all
  requests $r_{\pi(x)} \in V_{\starw_1(\tau)}$ we have $x \leq w$ and for all
  requests $r_{\pi(x)} \in V_{\starw_2(\tau)}$ we have $x \geq z$. To show
  this, we first observe that by the definition of $e$,
  $\starw_1(\tau)$ and $\starw_2(\tau)$, among all edges of
  $\starw(\tau)$, the edge $e=\set{r_{\pi(w)},r_{\pi(z)}}$ is the only edge that
  crosses the cut $(V_{\starw_1(\tau)},V_{\starw_2(\tau)})$. 

  We now first show that for all requests $r_{\pi(x)} \in V_{\starw_2(\tau)}$
  we have $x \geq z$. For contradiction, let us assume that there is a
  request $r_{\pi(x)} \in V_{\starw_2(\tau)}$ for which $x < z$ and therefore
  $x < w$. This implies that there must be a largest $y<w$ such that
  $r_{\pi(y)} \in V_{\starw_2(\tau)}$. Note that because $r_{\pi(y)}$ is not the
  last request ordered in $\tau$, $r_{\pi(\nextt(y))}$ must be in $\tau$
  and it therefore must be in $V_{\starw_1(\tau)}$. This implies that
  the edge $\set{r_{\pi(y)},r_{\pi(\nextt(y))}}$ of $\starw(\tau)$ crosses the cut
  $(V_{\starw_1(\tau},V_{\starw_2(\tau)})$, which is not possible
  because the edge $\set{r_{\pi(w)},r_{\pi(z)}}$ is the only edge of $\starw(\tau)$
  crossing this cut.

  We next show that for all requests $r_{\pi(x)} \in V_{\starw_1(\tau)}$, we
  have $x \leq w$. Again assume that there is a request
  $r_{\pi(x)} \in V_{\starw_1(\tau)}$ such that $x > w$ and thus $x > z$.
  Therefore, there must be smallest $y>w$ for which
  $r_{\pi(y)} \in V_{\starw_1(\tau)}$. This implies that $r_{\pi(y)}$ is the local
  successor of some request in $V_{\starw_2(\tau)}$.  This again
  contradicts the fact that the edge $e=\set{r_{\pi(w)},r_{\pi(z)}}$ is the only
  edge of $\starw(\tau)$ crossing the cut
  $(V_{\starw_1(\tau)},V_{\starw_2(\tau)})$.

  Finally we show that for all $r_{\pi(p)} \in V_{\starw_1(\tau)}$ and
  $r_{\pi(q)} \in V_{\starw_2(\tau)}$ the Manhattan cost of $e$ is at most
  $3 \cdot \mc(r_{\pi(p)},r_{\pi(q)})$.  Because $\pi$ is
  distance-respecting, we have
  \begin{equation}\label{eq:uparw1}
    t_{\pi(z)}-t_{\pi(q)} \leq d_{T}(v_{\pi(q)},v_{\pi(z)}) \leq \treelength(\lev).
  \end{equation}
  Further, if $p=0$, we have $t_{\pi(p)}=0$ and thus $t_{\pi(p)} \leq
  t_{\pi(w)}$. Otherwise, because $\pi$ is distance-respecting, we get
  \begin{equation}\label{eq:uparw2}
    t_{\pi(p)}-t_{\pi(w)} \leq d_{T}(v_{\pi(p)},v_{\pi(w)}) \leq \treelength(\lev).
  \end{equation}
  Using \eqref{eq:uparw1} and \eqref{eq:uparw2} we have
  \begin{equation}\label{eq:uparw3}
    t_{\pi(z)}-t_{\pi(w)} \leq t_{\pi(q)}-t_{\pi(p)} + 2 \cdot \treelength(\lev).
  \end{equation}
  
  We continue by distinguishing the two cases $t_{\pi(z)}\geq t_{\pi(w)}$ and
  $t_{\pi(w)}>t_{\pi(z)}$. First assume that $t_{\pi(z)} \geq t_{\pi(w)}$. Then, using
  $d_{T}(v_{\pi(z)},v_{\pi(w)})=d_{T}(v_{\pi(p)},v_{\pi(q)})=\treelength(\lev)$ and
  \eqref{eq:uparw3} we obtain
  \[
  \mc(r_{\pi(z)},r_{\pi(w)}) \leq \mc(r_{\pi(p)},r_{\pi(q)}) + 2 \cdot \treelength(\lev).	
  \]
  Moreover, because $d_{T}(v_{\pi(p)},v_{\pi(q)})=\treelength(\lev)$, we know
  that $\treelength(\lev) \leq \mc(r_{\pi(p)},r_{\pi(q)})$. Thus,
  \[
  \mc(e) \leq 3 \cdot \mc(r_{\pi(p)},r_{\pi(q)}).	
  \]
  Let us therefore consider the second case where $t_{\pi(w)} > t_{\pi(z)}$. It is
  clear that $w \neq 0$ as otherwise $t_{\pi(w)}=0$ and thus $t_{\pi(z)}\geq t_{\pi(w)}$.
  Because $\pi$ is distance-respecting, we have
  \[
  t_{\pi(w)}-t_{\pi(z)} \leq d_T(v_{\pi(w)},v_{\pi(z)}) = \treelength(\lev).
  \]
  Using the assumption that $t_{\pi(w)}>t_{\pi(z)}$, we then have
  \[
  \mc(r_{\pi(z)},r_{\pi(w)}) =|t_{\pi(w)}-t_{\pi(z)}|+d_{T}(v_{\pi(w)},v_{\pi(z)})  = t_{\pi(w)}-t_{\pi(z)} +
  \treelength(\lev) \leq 2 \cdot \treelength(\lev).	
  \]
  Finally, we can again use that $\mc(r_{\pi(p)},r_{\pi(q)}) \geq d_T(v_{\pi(p)},v_{\pi(q)}) =
  \treelength(\lev)$ and thus get that 
  \[
  \mc(e) \leq 2 \cdot \mc(r_{\pi(p)},r_{\pi(q)}).
  \]
  This concludes the proof of the lemma.
\end{proof}
\subsection{Bounding the Total Latency Cost}
\label{sec:generalframeworktotallatency}

It remains to prove the main claim and show that the total online
queueing cost on the HST $T$ is within a constant factor of the
optimal offline cost on $T$. The following theorem states that this is
generally true for algorithms with distance-respecting latency cost
(\Cref{def:distlatency}) and which produce distance-respecting
queueing orders (\Cref{def:distorder}), as long as the request set
$R$ is condensed (\Cref{def:condensed}).

\begin{theorem}\label{thm:totalonlineT}
  Assume that we are given an HST $T$ and a condensed set of requests issued at
  the leaves of $R$. Further, assume that we are given a distributed
  queueing algorithm $\ALG$ that has distance-respecting latency cost
  and that always produces a distance-respecting queueing order $\pi$. Then,
  the total latency cost of $\ALG$ is within a constant factor of the
  optimal offline cost on $T$.
\end{theorem} 
\begin{proof}
  Because the request set $R$ is condensed, \Cref{lemma:manhattanopt} implies that the optimal offline cost is
  within a constant factor of the Manhattan cost of an optimal TSP
  path connecting all the requests. The optimal offline cost therefore
  also is within a constant factor of the total Manhattan cost of an
  MST of the request set. Hence, \Cref{co:mstapprox} implies
  that also the total Manhattan cost of $S^*$ is within a constant
  factor of the cost of an optimal offline solution on $T$. Because
  the ordering $\pi$ generated by $\ALG$ is distance-respecting, by
  \Cref{le:arrowmstapprox}, the same is true for the total
  Manhattan cost $\mtntotalcost_{\mtn}(\starw)$ of the tree
  $\starw$. It therefore remains to show that
  $\cost_{\ALG}^T(\pi) = O(\mtntotalcost_{\mtn}(\starw))$.
  
  Because $\ALG$ has distance-respecting latency cost, for all $i
  \in[0,|R|-2]$, we have
  \[
  t_{\pi(i+1)}+\latency_{\ALG}^T(r_{\pi(i)},r_{\pi(i+1)})  \leq t_{\pi(\nextt(i))}+d_{T}(v_{\pi(i)},v_{\pi(\nextt(i))}).
  \]
  Note that we have $\nextt(i) \geq i+1$. Subtracting $t_{\pi(i)}$ on
  both sides yields
  \[
  t_{\pi(i+1)}-t_{\pi(i)}+ \latency_{\ALG}^T(r_{\pi(i)},r_{\pi(i+1)})
  \leq t_{\pi(\nextt(i))}-t_{\pi(i)}+d_{T}(v_{\pi(i)},v_{\pi(\nextt(i))}).
  \]
  If we sum up the above inequality for all $i\in[0,|R|-2]$, we get
  \[
  \sum_{i=0}^{|R|-2} \big(t_{\pi(i+1)}-t_{\pi(i)}+d_{T}(v_{\pi(i)},v_{\pi(i+1)})\big)
  \leq 
  \sum_{i=0}^{|R|-2} \big(t_{\pi(\nextt(i))}-t_{\pi(i)}+d_{T}(v_{\pi(i)},v_{\pi(\nextt(i))})\big)
  \]
  The sum of the latencies on the left-hand side exactly equals the
  total queueing cost $\cost_{\ALG}^T(\pi)$ of $\ALG$. To bound the
  right-hand side, note that
  $t_{\pi(\nextt(i))}-t_{\pi(i)}+d_{T}(v_{\pi(i)},v_{\pi(\nextt(i))})
  \leq \mc(r_{\pi(i)},r_{\pi(\nextt(i))})$. Together, we get
  \[
  t_{\pi(|R|-1)} - t_{\pi(0)} + \cost_{\ALG}^T(\pi) \leq \mtntotalcost_{\mtn}(\starw).
  \]
  As specified in \Cref{sec:model}, we assume that $t\geq 0$
  for every request $r=(v,t)$ and that every queueing algorithm first
  has to order the dummy request $r_0=(v_0,0)$. We therefore have
  $t_{\pi(|R|-1)}\geq 0$ and $t_{\pi(0)}=t_0=0$, which completes the
  proof of the theorem.
\end{proof}


\begin{corollary}\label{cor:synchArrowHST}
  The total latency cost of a synchronous execution of \arrow\ on an
  HST $T$ is within a constant factor of the optimal offline queueing
  cost on $T$.
\end{corollary}
\begin{proof}
  First note that by \Cref{le:transformation}, w.l.o.g., for
  synchronous \arrow\ executions, we can assume that the request set
  $R$ is condensed. The corollary therefore follows from \Cref{thm:totalonlineT} if we show that synchronous \arrow's ordering is
  distance-respecting and that synchronous \arrow\ has
  distance-respecting latency cost. The former follows from claim 2 of
  \Cref{le:timewindow}, the latter follows from claim 1 of \Cref{le:timewindow} and the fact that the latency cost of
  synchronous \arrow\ for ordering a request $r_i$ as the predecessor
  of request $r_{i+1}$ is exactly $d_T(v_i,v_{i+1})$.
\end{proof}

\begin{remark}
  The above corollary proves \Cref{thm:HSTmain} (cf.\ \Cref{sec:intro}) for synchronous executions on the HST $T$. The full
  statement of \Cref{thm:HSTmain} for general asynchronous
  executions is proven in \Cref{sec:asynchmodel}. There, it is
  shown that also for asynchronous executions, \arrow\ has
  distance-respecting latency cost and produces distance-respecting
  queueing orders. In addition, we also show that we can still
  restrict attention to condensed request sets. The claim of \Cref{thm:HSTmain} for the asynchronous case then follows from \Cref{thm:totalonlineT} in the same way as in the above corollary.
\end{remark}



\section{Queueing Cost in the Asynchronous Model}
\label{sec:asynchmodel}

In this section, we show that the generic analysis of \Cref{sec:arrowanalysis} also applies to asynchronous executions of the
\arrow\ protocol on $T$. In order to use the framework of \Cref{sec:arrowanalysis} in the asynchronous setting, we mostly
importantly need to show that \arrow\ has distance-respecting latency
cost (\Cref{def:distlatency}) and that it generates
distance-respecting queueing orders (\Cref{def:distorder}) also in
the asynchronous case. To show this, we need asynchronous variants of
the basic \Cref{le:transformation} and \Cref{le:timewindow}. In
addition, we also need to generalize \Cref{le:transformation} to
show that also in the asynchronous setting, w.l.o.g., we can assume
that the given request set $R$ is condensed
(\Cref{def:condensed}).

As in \Cref{sec:treeanalysis}, we relabel the requests for
convenience. Throughout the section, we assume that an asynchronous
execution $\asyorder$ of \arrow\ is given and we label the requests
according the order $\asyorder$. That is, $r_0$ is the dummy request
and for every $i\geq 1$, $r_i$ is the $i^{\mathit{th}}$ non-dummy
request ordered by the asynchronous \arrow\ execution.

\subsection{Basic Properties of Asynchronous Arrow Executions}
\label{sec:asynchbasic}

We have seen that a synchronous \arrow\ execution can be seen as a
greedy queueing order in the following sense. Assume that requests
$r_0,\dots,r_{i-1}$ of the queueing order are known and let $v_{i-1}$
be the node at which request $r_{i-1}$ has been issued. Then, request
$r_i$ the first one among the remaining requests that reaches node
$v_{i-1}$ on a direct path. In the asynchronous setting, an analogous
property is true. However, we need to be a bit more careful and argue
the arrival time of the ``find predecessor'' message on the
whole path from the node of a request to its predecessor.

Let us assume that we are given a tree $T$, a dynamic set of
requests $R$ issued at the nodes of $T$, as well as an asynchronous
execution of \arrow\ that orders the requests
$r_0,r_1,\dots,r_{|R|-1}$ in this order. It has been shown in \cite{demmer1998arrow} that
even in a concurrent asynchronous \arrow\ execution, every request $r_i$
finds the node $v_{i-1}$ of its predecessor $r_{i-1}$ on a direct path. To formally
specify the greedy property of \arrow\ in the asynchronous setting, we
need to study the progress of messages on the whole path from a
request to its predecessor. For any two nodes $u,v$ of $T$, we use
$P_{u,v}$ to denote the direct path from $u$ to $v$ on tree $T$.  The
following \Cref{lemma:asynchgreedy} formally establishes the
greedy behavior of asynchronous \arrow\ executions.

We first introduce some terminology defined in
\cite{herlihy2006dynamic}.  For all $i\in [0,|R|-1]$, we define $F_i$
to be a configuration of the tree network, where all arrows are
pointing towards the node $v_i$ of request
$r_i$. Further, let $R_i$ be the set of requests
$[r_{i+1},|R|-1]$ that are ordered after request
$r_i$. Finally, let $E_i$ be an execution of the \arrow\
protocol starting from configuration $F_i$ and in which only the
requests in $R_i$ are issued. It is shown in Lemma 3.7 in
\cite{herlihy2006dynamic} that for all $i$, except for request
$r_i$ no request in $R_{i-1}$ can distinguish locally
between executions $E_{i-1}$ and $E_i$. More specifically, all these
requests see exactly the same arrows in both executions. This implies
that the ``find predecessor'' message of every request
$r_i$ sees exactly the same arrows as if the network
started in configuratoin $F_{i-1}$ and only request $r_i$
was issued. To study the behavior of the requests in $R_{i-1}$, it
therefore suffices to study an execution that starts in configuration
$F_{i-1}$ and where only the requests in $R_{i-1}$ are issued.

\begin{lemma}\label{lemma:asynchgreedy}
  Consider an asynchronous \arrow\ execution for a request set $R$ on
  a tree $T$. Let $i\in [1,|R|-1]$ and consider the path
  $P_{v_i,v_{i-1}}=(u_0,u_1,\dots,u_s)$ from node $u_0=v_{i}$ of
  request $r_i$ to the node $u_s=v_{i-1}$ of the predecessor
  $r_{i-1}$. For every node $u_k$ on the path, the ``find
  predecessor'' message of request $r_{i}$ is the first ``find
  predecessor'' message that reaches node $u_k$ (or is generated at
  node $u_k$) among all the ``find predecessor'' of requests $r_j$ for
  $j\in[i,|R|-1]$.
\end{lemma}
\begin{proof}
  In order to prove the claim of the lemma, we can assume that
  requests $r_0,\dots,r_{i-1}$ have already found their predecessors
  and therefore the tree is in configuration $F_{i-1}$. Lemma 3.7 in
  \cite{herlihy2006dynamic} implies that this does not affect the
  behavior of any of the remaining queueing requests in $R_{i-1}$.

  Assume for contradiction that the claim of the lemma is not
  true. Let $x\in[0,\dots,s]$ be the maximal value such that the
  ``find predecessor'' message of request $r_{i}$ is not
  the first one among the requests in $R_{i-1}$ reaching
  $u_k$. Note that we need to have $k<s$ because by the definition of
  the \arrow\ protocol, the first message
  reaching $u_s=v_{i-1}$ is the successor request of
  $r_{i-1}$. Let $r=(v,t)$ be the first request in $R_{i-1}$ that reaches node $u_s$. In
  configuration $F_{i-1}$, the arrow of node $u_k$
  points to $u_{k+1}$. In order to change this, a ``find predecessor''
  message first has to be sent from node $u_k$ to $u_{k+1}$. Because
  $r$ is the first request reaching $u_k$, when the ``find
  predecessor'' message of $r$ reaches $u_k$, this has not happened
  and therefore the arrow still points from $u_k$ to $u_{k+1}$. When
  reaching $u_k$, in an atomic step, the ``find predecessor'' message
  of $r$ is therefore forwarded to $u_{k+1}$. As long as the message
  is in transit between the two nodes, there is no arrow across the
  edge $\set{u_k,u_{k+1}}$ and therefore the ``find predecessor''
  message of $r$ also reaches $u_{k+1}$ before the ``find
  predecessor'' message of $r_{i}$ reaches $u_{k+1}$. This
  is a contradiction to the assumption on the maximality of $k$ and
  therefore the claim of the lemma holds.
\end{proof}

The above lemma shows that if the ``find predecessor'' messages of two
requests reach the same node $v$, then the earlier ordered request
reaches $v$ first. To have an analogous statement for \Cref{le:timewindow}, we would like to have a statement saying that a
request $r$ reaches a node $v$ on the path to the predecessor request
before any request $r'$ that is ordered after $r$ (not only for a
request $r'$ that actually reaches $v$). To achieve this, we extend a
given \arrow\ execution to simplify the analysis. Whenever a request
$r=(v,t)$ is issued at node $v$ at time $t$, a ``find predecessor''
message leaves $v$ at time $t$ and it travels on the direct path to
the predecessor request $r'$ of $r$. For the proof, we assume that
instead of only going to the predecessor, the
``find predecessor'' message is sent as a broadcast to the whole
network. We think of the additional messages to complete this
broadcast as virtual
messages that are only used for the analysis and have no influence on
the queueing protocol. Given an asynchronous execution of \arrow, we
assume that the actual messages sent by the \arrow\ protocol keep
their message delays (to ensure an equivalent execution). All the
virtual messages are assumed to have the maximum possible message
delay. That is, the delay of sending a virtual message from $u$ to $v$
is equal to the length $d_T(u,v)$ of the respective tree
edge. Further, to make sure that virtual messages can never overtake
real messages, if a real message and a virtual message reach a node at
the same time, the node always first processes the real message. In
this way, for every request $r=(v,t)$, the delay of the respective
``find predecessor'' message is defined for all nodes. For a request
$r$ and a node $u\in V$, we introduce the following notation:
\begin{equation}\label{eq:asynchdelay}
  \Delta(r,u)\ :=\ \text{time of ``find predecessor'' message of
    request $r$ to reach node $u$}.
\end{equation}
We note that for $r=(v,t)$ and any node $u\in V$, we have
$\Delta(r,u)\leq d_T(u,v)$ (recall that in the asynchronous setting,
for the analysis, the delay of a message is assumed to be at most the
length of the respective edge). The next lemma will be used as a
replacement of the main statement of \Cref{le:timewindow} in the
asynchronous analysis.

\begin{lemma}\label{le:asynchtimewindow}
  Consider an asynchronous execution of \arrow\ for a set of requests
  $R$ on tree $T$ and consider
  two arbitrary requests $r_{i}$ and $r_{j}$ for which $1\leq i<j$ (i.e.,
  $r_{j}$ is ordered after $r_{i}$ by
  \arrow). Then for any node $v$ on the path from $v_i$
  to $v_{i-1}$, it holds that
  \[
  t_i + \Delta(r_{i},v) \leq t_j + \Delta(r_{j},v).
  \]
\end{lemma}
\begin{proof}
  Similarly to the proof of \Cref{lemma:asynchgreedy}, we apply
  Lemma 3.7 from \cite{herlihy2006dynamic} and we assume that the
  network starts in configuration $F_{i-1}$. Consequently, initially,
  all arrows are pointing towards $v_{i-1}$ and only the requests in
  $R_{i-1}$ still need to be ordered.

   We first show that for every arrow pointing from a node $u_1$ to a
   node $u_2$ in configuration $F_{i-1}$, the first message sent from
   $u_1$ to $u_2$ has to be a real message. For contradiction, assume
   otherwise and assume that the first arrow along which a virtual
   message is sent before a real message is pointing from node $w_1$
   to node $w_2$. Further, assume that message $\mathcal{M}$ is the
   first such message that is sent by $w_1$ over the edge. Note that
   this also implies that $\mathcal{M}$ is the first message sent from
   $w_1$ to $w_2$. Assume that this virtual message $\mathcal{M}$
   belongs to a request $r=(v,t)$. First note that $\mathcal{M}$ is
   the first message arriving at node $w_1$. Otherwise, some other
   message would have been sent from $w_1$ to $w_2$. If message
   $\mathcal{M}$ arrives at $w_1$ as a real message, it is forwarded
   as a real message to node $w_2$. We can therefore conclude that
   message $\mathcal{M}$ reaches $w_1$ as a virtual message (say from
   neighbor $w_0$). Because $\mathcal{M}$ is the first message that
   reaches $w_1$, it is also the first message sent from $w_0$ to
   $w_1$ (note that as a virtual message, it has the maximum possible
   message delay, so it cannot overtake any other message). Because in
   configuration $F_{i-1}$, there also is an arrow from $w_0$ to
   $w_1$, this is a contradiction to the assumption that the arrow
   from $w_1$ to $w_2$ is the first on which a virtual message is sent
   before a real one.

   To conclude the proof, observe that in configuration $F_{i-1}$, all
   neighbors $u$ of the path $P_{v_{i},v_{i-1}}=(u_0,\dots,u_s)$ from
   $u_0=v_{i}$ to $u_s=v_{i-1}$ have an arrow pointing from $u$ to the
   neighbor on the path. Hence, on each edge connecting to the path,
   the first message that reaches the path is a real message. The same
   is true for all edges of the path in the direction from node
   $u_0=v_{i}$ to node $u_s=v_{i-1}$. The only way a virtual message
   can therefore reach a node $u_k$ of the path before a real message
   does is when a virtual message for a request $r$ is sent from a
   node $u_{k+1}$ to node $u_k$. Assume that this is the case and
   assume that $u_x$ for $x\geq k+1$ is the first node on the path
   that is reached by the message of $r$. There are two cases to
   consider, either the message of $r$ reaches node $u_x$ from a
   neighbor outside the path $P_{v_{i},v_{i-1}}$ or the request is
   issued at node $u_x$. Because the first message reaching the path
   $P_{v_i,v_{i-1}}$ from a neighbor of the path has to be a real
   message, \Cref{lemma:asynchgreedy} implies that the ``find
   predecessor'' message of request $r_i$ reaches $u_x$ before any
   message from outside the path reaches $u_x$. However, in that case,
   the ``find predecessor'' message of $r_i$ also reaches all earlier
   nodes on path $P_{v_i,v_{i-1}}$ (and thus in particular node $u_k$)
   before the message of $r$ does. If the request $r$ is issued at
   node $u_x$, \Cref{lemma:asynchgreedy} also implies that this
   has to happen after the ``find predecessor'' message of $u_i$
   reaches $u_x$.
\end{proof}

The following lemma is a simple consequence of \Cref{le:asynchtimewindow}.

\begin{lemma}\label{lemma:distancerespecting}
  Consider an asynchronous execution of \arrow\ for a given set of
  requests $R$ on a tree $T$ and consider two arbitrary requests $r_i$
  and $r_j$ for which $i<j$ (i.e., $r_i$ is ordered before
  $r_j$). Then, the following two statements hold:
  \begin{enumerate}
  \item $t_i - t_j \leq d_T(v_i,v_j)$,
  \item if $i\geq 1$, $t_i + \Delta(r_i,v_{i-1}) \leq t_j + d_T(v_{i-1},v_j)$.
  \end{enumerate}
\end{lemma}
\begin{proof}
  If $i=0$, we only need to prove the first claim, which in this
  clearly holds because $t_0=0$ and $t_j\geq 0$ for all $r_j\in
  R$.
  Let us therefore assume that $i\geq 1$. We consider the part of the
  tree $T$ induced by the paths between the nodes $v_i$, $v_j$, and
  the node $v_{i-1}$ of the predecessor request $r_{i-1}$ of
  $r_i$. Let $x$ be the (unique) node on the tree on which the three
  paths $P_{v_i,v_j}$, $P_{v_i,v_{i-1}}$, and $P_{v_j,v_{i-1}}$
  intersect. Because $x$ in particular is a node on the path
  $P_{v_i,v_{i-1}}$, from \Cref{le:asynchtimewindow}, we get that
  \begin{equation}\label{eq:applyasynchtimewindow}
    t_i + \Delta(r_i,x) \leq t_j + \Delta(r_j,x).
  \end{equation}
  The term $\Delta(r_j,x)$ is the delay of the message of
  request $r_j$ to reach node $x$ from node $v_j$. Because the message
  delay is upper bounded by the length of the path and because $x$ is on
  the path $P_{v_i,v_j}$, we have $\Delta(r_j,x)\leq d_T(v_j,x)\leq
  d_T(v_j,v_i)$ and thus, the first claim of the lemma follows
  directly from \eqref{eq:applyasynchtimewindow} (note
  that $\Delta(r_i,x)\geq 0$). The second claim can also be proved
  based on \eqref{eq:applyasynchtimewindow}:
  \begin{eqnarray*}
    t_i + \Delta(r_i, v_{i-1})
    & = &
          t_i + \Delta(r_i, x) + \big(\Delta(r_i,v_{i-1}) -
          \Delta(r_i,x)\big)\\
    & \stackrel{\eqref{eq:applyasynchtimewindow}}{\leq} &
             t_j + \Delta(r_j,x) + \big(\Delta(r_i,v_{i-1}) -
          \Delta(r_i,x)\big)\\
    & \leq &
             t_j + d_T(v_j,x) + d_T(x, v_{i-1})\\
    & = &
          t_j + d_T(v_{i-1},v_j).
  \end{eqnarray*}
  The second inequality follows because the message delay of an edge
  is at most the length of the edge.
\end{proof}

It remains to adapt the basic \Cref{le:transformation} to the
asynchronous setting.

\begin{lemma}\label{le:asynchtransformation}
  Let $R$ be a set of queueing requests issued on a tree $T$ and let
  $r_i=(v_i,t_i)$ and $r_j=(v_j,t_j)$ be two requests of $R$ that are
  consecutive w.r.t.\ time of occurrence. Further, choose two requests
  $r_a=(v_a,t_a)$ with $t_a\leq t_i$ and $r_b=(v_b,t_b)$ with
  $t_b\geq t_j$ minimizing $\delta:=t_b-t_a-d_T(v_a,v_b)$. If
  $\delta>0$, every request $r=(v,t)$ with $t\geq t_j$ can be replaced
  by a request $r'=(v,t-\delta)$ without decreasing the worst-case
  cost of \arrow\ and without increasing the optimal offline cost.
\end{lemma}
\begin{proof}
  Because the optimal offline cost is computed w.r.t.\ synchronous
  executions, the proof that the optimal offline cost is not increased
  follows directly from \Cref{le:transformation}. To show that the
  worst-case \arrow\ cost does not decrease, we show that if all the
  message delays remain the same, the execution can still produce the
  same \arrow\ order with the same total cost.

  Let $R_{\leq}$ be the set of requests with issue time $\leq t_i$ and
  let $R_{\geq}$ be the set of requests with issue time $\geq t_j$.
  Note that $R=R_{\leq} \cup R_{\geq}$. We first show that when
  replacing every request $r=(v,t)$ in $R_{\geq}$ by a request
  $r'=(v,t-\delta+\eps)$ for an arbitrary $\eps>0$, if we do not
  change any of the message delays, we obtain exactly the same \arrow\
  ordering and cost.\footnote{A bit more precisely, the asynchronous
    scheduler has to generate the same message delays and whenever
    several messages arrive at some node at exactly the same time, the
    scheduler needs to process them in the same order.} To see this,
  first observe that in this case, claim 1 of \Cref{lemma:distancerespecting} implies that all requests in
  $R_{\leq}$ are ordered before any request in $R_{\geq}$ is
  ordered. Let $r_x=(v_x,t_x)$ be the last request ordered in
  $R_{\leq}$ and let $r_y=(v_y,t_y)$ be the first request ordered in
  $R_{\geq}$ in the original execution. Because all requests in
  $R_{\geq}$ are shifted by the same amount and they are still all
  ordered after the requests in $R_{\leq}$, also after the shifting,
  the ``find predecessor'' request of $r_y$ is the first one to arrive
  at node $v_x$ and therefore $r_y$ still is the successor of
  $r_x$. Because the time differences inside $R_{\geq}$ do not change,
  also the rest of the ordering does not change. The argument holds
  even if we let $\eps$ go arbitrarily close to $0$. In the limit, the
  argument therefore still holds as long as whenever a node receives
  several messages at the same time, the asynchronous scheduler 
  processes messages corresponding to requests in $R_{\leq}$ before
  processing messages corresponding to $R_{\geq}$. We have therefore
  shown that for every initial \arrow\ execution, the asynchronous
  scheduler can enforce an equivalent execution with the same cost
  with the shifted request. This proves the claim of the lemma.
\end{proof}

We now have everything needed to prove \Cref{thm:HSTmain}
stating that the total cost of an asynchronous execution of \arrow\ on
an HST $T$ is within a constant factor of the optimal offline queueing
cost on $T$.

\begin{proof}[\textbf{Proof of \Cref{thm:HSTmain}}]
  The above \Cref{le:asynchtransformation} shows that we can
  (iteratively) transform the initial request set $R$ into a condensed set of
  requests without decreasing the cost of \arrow\ and without
  increasing the optimal offline cost. We can therefore assume that we
  are given a condensed set of requests. The claim of the theorem now
  follows if we can show that the latency cost of asynchronous \arrow\
  is distance-respecting and that any asynchronous \arrow\ execution
  generates a distance-respecting queueing order. However, these
  statements follow directly from claims 2 and 1 of
  \Cref{lemma:distancerespecting}, respectively.
\end{proof}


\bibliographystyle{abbrv}
\bibliography{references}

\newpage
\appendix

\section{Minimum Spanning Tree Approximation}
\label{sec:MSTapprox}

In this section, we prove a general minimum spanning tree (MST)
approximation result. Assume that we are given a spanning tree $\tau=(V,E_\tau)$ of a
graph $G=(V,E)$. Together with $\tau$, every edge $e\in E_\tau$
induced a cut of $G$ as follows. When removing $e$ from $\tau$, we
obtain a spanning forest consisting of two connected subtrees of
$\tau$. Let $S$ and $V\setminus S$ be the node sets of these two
connected components. We say that $(S,V\setminus S)$ is the \emph{cut
induced by removing $e$ from $\tau$}. The next theorem shows that
if for every edge $e\in E_\tau$, the weight of $e$ is within a factor
$\lambda$ of the weight of the lightest edge crossing the cut induced
by removing $e$ from $\tau$, then the total weight of $\tau$ is within a factor
$\lambda$ of the weight of an MST. We expect that this results is
already known, however, we have not found a proof of it in the
literature. The next theorem proves a slightly more general statement.

\begin{theorem}\label{thm:MSTapprox}
  Let $\lambda\geq 1$ be some number and let $G=(V,E,w)$ be a weighted
  connected graph with non-negative edge weights $w(e)\geq 0$ and let
  $\tau\subseteq E$ and $\tau^*\subseteq E$ be two arbitrary spanning
  trees of $G$. If for every edge $e$ of $\tau$, the lightest edge
  $e'$ of $\tau^*$ crossing the cut induced by removing $e$ from
  $\tau$ has weight $w(e')\geq w(e)/\lambda$, then the total weight of
  all edges in $\tau$ is at most a $\lambda$-factor larger than the
  total weight of the edges in $\tau^*$.
\end{theorem}
\begin{proof}
  In the following, we slightly abuse notation and we identify a
  spanning tree $\tau$ with the set of edges contained in $\tau$. For
  an edge set $F\subseteq E$, we also use $w(F)$ to denote the total
  weight of the edges in $F$. We prove the stronger statement that
  \begin{equation}\label{eq:MSTapprox}
    w(\tau\setminus \tau^*) \leq \lambda\cdot 
    w(\tau^*\setminus \tau).
  \end{equation}
  We show \eqref{eq:MSTapprox} by induction on
  $|\tau\setminus \tau^*| = |\tau^*\setminus \tau|$. First note that
  if $|\tau\setminus \tau^*| = 0$, we have $\tau=\tau^*$ and thus
  \eqref{eq:MSTapprox} is clearly true. Further, if
  $|\tau\setminus \tau^*|=1$, there is exactly one edge
  $e\in \tau\setminus \tau^*$ and exactly one edge
  $f\in \tau^*\setminus \tau$. Because $\tau$ and $\tau^*$ are
  spanning trees, $f$ connects the two sides of the cut
  $(V_{e,1},V_{e,2})$ induced by removing $e$ from $\tau$ and we therefore have
  $w(f)\leq \lambda \cdot w(e)$, implying \eqref{eq:MSTapprox}.

  Let us therefore assume that $|\tau\setminus \tau^*|=k\geq 2$ and
  let $e$ be a maximum weight edge of $\tau\setminus \tau^*$. Let
  $(V_{e,1},V_{e,2})$ be the cut induced by removing $e$ from $\tau$. Further,
  let $\tau'$ be a spanning tree of $G$ that is obtained by removing
  $e$ from $\tau$ and by adding some edge $f\in \tau^*\setminus \tau$
  that connects $V_{e,1}$ and $V_{e,2}$. Note that by the assumptions
  of the theorem, we have $w(e)\leq \lambda\cdot w(f)$. To prove
  \eqref{eq:MSTapprox}, it thus suffices to show that
  $w(\tau'\setminus \tau^*)\leq \lambda\cdot w(\tau^*\setminus
  \tau')$.
  We have $|\tau'\setminus \tau^*|=k-1$ and thus, if the spanning tree
  $\tau'$ satisfies the conditions of the theorem,
  $w(\tau'\setminus \tau^*)\leq \lambda\cdot w(\tau^*\setminus \tau')$
  and \eqref{eq:MSTapprox} follows from the induction
  hypothesis. We therefore need to show that $\tau'$ satisfies the
  conditions of the theorem.

  Consider an arbitrary edge $e'\in \tau'\setminus \tau^*$ and let
  $(U_1,U_2)$ be the partition of $V$ induced by removing $e'$ from
  tree $\tau'$. Since $e'$ is an
  edge of one of the two subtrees of $\tau$ resulting after removing $e$,
  $e'$ either connects two nodes in $V_{e,1}$ or two nodes in
  $V_{e,2}$. W.l.o.g., assume that $e'$ connects two nodes in
  $V_{e,2}$ and let $V_{e,2,1}$ and $V_{e,2,2}$ be the partition of
  $V_{e,2}$ induced by removing $e'$ from the subtree of $\tau$
  induced by $V_{e,2}$. We need to show that for every edge $f'\in
  \tau^*$ connecting $U_1$ and $U_2$, it holds that $w(e')\leq
  \lambda\cdot w(f')$. Any edge $f'$ crossing the cut has to either
  connect $V_{e,1}$ with $V_{e,2}$ or it has to connect $V_{e,2,1}$
  with $V_{e,2,2}$. In the first case, we have $w(e')\leq w(e)\leq
  \lambda\cdot w(f')$ (recall that we chose $e$ to be the heaviest
  edge from $\tau\setminus \tau^*$). In the second case, $f'$ also crosses
  the cut induced by removing $e'$ from the original tree $\tau$ and
  therefore we also have $w(e')\leq \lambda\cdot w(f')$. This
  concludes the proof.
\end{proof}


\end{document}